\documentclass{amsart}
\usepackage[draft]{marktext}

\usepackage[english]{babel}
\usepackage[utf8]{inputenc}

\usepackage{algpseudocode, algorithm}

\usepackage[draft]{gimac}

\newcommand{\BV}{\mathrm{BV}}

\newcommand{\osc}{\mathrm{osc}}

\usepackage{colortbl}
\colorlet{tableHeader}{cyan!5}

\newtheorem{assumption}[theorem]{Assumption}

\newcounter{broj}

\DeclareMathOperator{\PDF}{PDF}
\DeclareMathOperator{\Hess}{Hess}

\title[Sampling from multimodal distributions on~$\mathbb{R}^d$]{Convergence of a Sequential Monte Carlo algorithm towards multimodal distributions on~$\mathbb{R}^d$}

\begin{document}

\begin{abstract}
In an earlier joint work, we studied a sequential Monte Carlo algorithm to sample from the
Gibbs measure supported on torus with a non-convex energy function at a low temperature, where we proved that the time complexity of the algorithm is polynomial in the inverse temperature. However, the analysis in that torus setting relied crucially on compactness and does not directly extend to unbounded domains. This work introduces a new approach that resolves this issue and establishes a similar result for sampling from Gibbs measures supported on~$\R^d$.
  In particular, our main result shows that for double-well energy with equal well depths, the time complexity scales  as seventh power of the inverse temperature, and quadratically in both the inverse allowed absolute error and probability error.
\end{abstract}

\author[Han]{Ruiyu Han}
\address{%
  Department of Mathematical Sciences, Carnegie Mellon University, Pittsburgh, PA 15213.
}
\email{ruiyuh@andrew.cmu.edu}

\subjclass{%
  Primary:
    60J22,  
  Secondary:
    65C05, 
    65C40, 
    60J05, 
    60K35. 
  }
\keywords{
  Markov Chain Monte Carlo,
  sequential Monte Carlo,
  annealing,
  multimodal distributions,
  high dimensional sampling.
}

\maketitle

\section{Introduction}

We show that under general non-degeneracy conditions, the Annealed Sequential Monte Carlo algorithm (Algorithm \ref{a:ASMC}) can efficiently sample from multimodal distributions on~$\R^d$, with time complexity that is polynomial in the inverse temperature, with a precise, dimension-independent degree. In previous work with Iyer and Slep\v{c}ev \cite{han2025polynomialcomplexitysamplingmultimodal}, we proved such a result for distributions supported on~$\mathbb{T}^d$, using arguments specific to the torus setting.
 The main purpose of this paper is to obtain a similar result in the non-compact setting of~$\R^d$.
 Our main result shows that the time complexity of ASMC in this settings is also polynomial in the inverse temperature, albeit with a slightly larger, but precise and dimension-independent degree.
 There are several key steps in the proof in~\cite{han2025polynomialcomplexitysamplingmultimodal} which can not be used in the non-compact setting.  We overcome these challenges here by using a different method based on a coupling argument in the work of Marion, Mathews and Schmidler~\cite{Marion23SMC}.
 
We begin (Section~\ref{s:informal}) with an informal description of the algorithm and our results.
In Section~\ref{s:litReview}, we briefly review the related work.
Our main theorem and precise assumptions are stated in Section~\ref{sec:main}.  The remainder of this paper is devoted to the proofs.

\subsection{Informal statement of the main results}\label{s:informal}
Let~$U \colon \mathcal{X}\to \R$ be a non-convex energy function
defined on a space~$\mathcal{X}$. In this paper the space X will 
be the $d$-dimensional Euclidean space $\R^d$.
Consider the Gibbs distribution~$\pi_\epsilon$ whose density is given by
\begin{equation}\label{e:piNu}
  \pi_{\epsilon}(x) = \frac{1}{Z_\epsilon} \tilde \pi_\epsilon(x)
  ,
  \quad\text{where}\quad
  \tilde \pi_\epsilon(x) \defeq e^{-U(x) / \epsilon}
  \text{ and }
  Z_\epsilon \defeq \int_{\mathcal X} \tilde \pi_\epsilon(y) \, d y
  .
\end{equation}
where~$dy$ denotes some fixed measure on the configuration space~$\mathcal{X}$. The parameter~$\epsilon > 0$ denotes the temperature. Sampling from~$\pi_{\epsilon}$ in the regime of low temperature~$\epsilon$ is an important problem in multiple disciplines,
such as Bayesian statistical inference, machine learning and statistical physics~\cite{BirderHeermann,DelDJ,Faulkner_2024}. Sampling from multimodal distributions is challenging.  \emph{Sequential Monte Carlo (SMC)} algorithm (see for instance~\cite[Chapters 3.3]{ChopinPapaspiliopoulos20}, or~\cite[Chapter 3.4]{Liu08}) and its variants~\cite{Neal01,popMC,popanneal}, along with related tempering methods~\cite{Swendsen1986tempering,geyer1991markov,MarinariParisi92,Geyer1995}, are designed to solve those problems by combining Markov
chain Monte Carlo (MCMC) methods and resampling strategies to sequentially sample from
a sequence of probability distributions.


The \emph{Annealed Sequential Monte Carlo (ASMC)} algorithm is a SMC algorithm, where particles are moved through a series of interpolating measures obtained by gradually reducing the temperature according to a specified annealing schedule. As in \cite{han2025polynomialcomplexitysamplingmultimodal}, we use the annealing schedule where the inverse temperatures are linearly spaced~\cite{SyedBouchardCoteEA24}.
We now precisely state the ASMC algorithm used in~\cite{han2025polynomialcomplexitysamplingmultimodal} to sample from target distribution~$\pi_{\eta}$.

\begin{algorithm}[htb]
  \caption{Annealed Sequential Monte Carlo (ASMC) to sample from~$\pi_\eta$ \cite{han2025polynomialcomplexitysamplingmultimodal}.}
  \label{a:ASMC}
  \begin{algorithmic}[1]\reqnomode
    \Require Temperature~$\eta$, energy function~$U$, and Markov processes~$\set{Y_{\epsilon, \cdot}}_{\epsilon \geq \eta}$ so that the stationary distribution of~$Y_{\epsilon, \cdot}$ is~$\pi_\epsilon$.
    \item[\textbf{Tunable parameters:}]
    \Statex
      \begin{enumerate}[(1)]
        \item Number of levels~$M \in \N$, and annealing schedule~$\eta_1 > \cdots > \eta_M = \eta$.
	\item Sample size~$N \in \N$, and initial points~$y^1_1$, \dots,~$y^N_1 \in \mathcal X$.
	\item Level running time~$T > 0$.
      \end{enumerate}

    \For{$k \in \set{1, \dots, M-1}$}

      \State\label{i:exploreValley}
	For each~$i \in \set{1, \dots N}$, simulate~$Y_{\eta_k, \cdot}$ for time~$T$ starting at~$y^i_k$ to obtain~$x^i_k$.

      \State\label{i:resampling}
	Choose~$(y^1_{k+1}, \dots, y^N_{k+1})$ by resampling from~$\set{x^1_k, \dots, x^N_k}$ using
	the multinomial distribution with probabilities 
\begin{equation}\label{e:rkTilde}
  \P( y^i_{k+1} = x^j_k ) = \frac{\tilde r_k(x^j_k)}{\sum_{n=1}^N \tilde r_k(x^n_k)}
  ,\quad \tilde r_k
    \defeq \frac{\tilde \pi_{\eta_{k+1}}}{\tilde \pi_{\eta_k}}.
\end{equation}

    \EndFor

    \State
      For each~$i \in \set{1, \dots N}$, simulate~$Y_{\eta_M, \cdot}$ for time~$T$ starting at~$y^i_M$ to obtain~$x^i$.
    \State\label{i:RTlast}%
      \Return~$(x^1, \dots, x^N)$.
  \end{algorithmic}
\end{algorithm}

 A detailed discussion of the intuition and motivation behind this algorithm can be found in our previous work \cite{han2025polynomialcomplexitysamplingmultimodal}. The results in \cite{han2025polynomialcomplexitysamplingmultimodal} bound the~$L^2$ error of the Monte Carlo when the distribution is supported on~$\T^d$. Here our main theorem controls the Monte Carlo error with high probability in the non-compact setting~$\mathbb{R}^d$. As we will see shortly, the difference in the main results is caused by the unboundedness of the domain.
 An informal version of our main result is provided below and the precise theorem is stated as Theorem~\ref{thm: main} in Section~\ref{sec:main}.  
\begin{theorem}\label{t:mainIntro}
  Suppose~$U \colon \R^d \to \R$ is a non-degenerate double-well function with wells of equal depth (but not necessarily the same shape).
  For~$\epsilon > 0$ let~$Y_{\epsilon, \cdot}$ be a solution to the overdamped Langevin equation
  \begin{equation}\label{e:Langevin}
    dY_{\epsilon,t} = - \grad U(Y_{\epsilon,t}) \, dt + \sqrt{2 \epsilon} \, dW_t
    ,
  \end{equation}
  where~$W$ is a standard~$d$-dimensional Brownian motion.
  There exist constants~$C_N$ and~$C_T$, depending on~$U$ and~$d$, such that  the following holds.
  For any~$\delta > 0$,~$\eta > 0$ and~$\theta\in (0,1)$, choose~$M, N, T$ according to
  \begin{align}
    M &\geq \ceil*{ \frac{1}{\eta} },
    \quad 
    N \geq \frac{C_{N} M^2}{\delta^2}\log\paren[\Big]{\frac{M}{\theta}},\\
    T&\geq C_T\paren[\Big]{ \paren[\Big]{\frac{MN}{\theta}}^{\frac{4}{3}}\paren[\Big]{\log N+\log\paren[\Big]{\frac{M}{\theta}} }+\log\paren[\Big]{\frac{1}{\delta}}}
  \end{align}
  and a suitable geometric annealing schedule~$\set{ 1/\eta_k}_{k = 1, \dots, M}$ so that~$\eta_1$ is sufficiently large, and~$\eta_M = \eta$.
  Then the points~$x^1$, \dots,~$x^N$ obtained from Algorithm~\ref{a:ASMC} are such that for any bounded test function~$h$ we have
  \begin{equation}
      \P\paren[\bigg]{\abs[\Big]{
	\frac{1}{N}\sum_{i=1}^{N}h(x^i)-\int h(x)\pi_{\eta}(x)\,d x}
	  < \|h\|_{\osc} \delta}\geq 1-\theta
	.
  \end{equation}
\end{theorem}

Note that the drift in~\eqref{e:Langevin} is independent of temperature $\epsilon$, and so time complexity of the algorithm is proportional to~$MNT$. Thus Theorem~\ref{t:mainIntro} shows that the time complexity of obtaining good samples from~$\pi_{\eta}$ using ASMC is polynomial in~$1/\eta$, with a degree independent of dimension.
In contrast, the time complexity of obtaining good samples by directly simulating the process~$Y_{\eta, \cdot}$ is~$e^{O(1/\eta)}$, and the time complexity of importance sampling or rejection sampling (from $\eta_1$ directly to $\eta$) is $(1/\eta)^d$.

We now briefly explain the main difference between Theorem~\ref{t:mainIntro} and the results in our previous work~\cite{han2025polynomialcomplexitysamplingmultimodal}. As we mentioned earlier, the result in~\cite{han2025polynomialcomplexitysamplingmultimodal} controls the $L^2$ error while here we have high probability results. The reason for this is that our analysis relies on a eigenfunction bound that blows up if the particles drift arbitrarily far from modes of the distribution. In particular, a uniform bound on the~$L^{\infty}$ norm of the normalized eigenfunctions is crucial in the argument in~\cite{han2025polynomialcomplexitysamplingmultimodal}. However, in the~$\mathbb{R}^d$ case, eigenfunctions (other than the first eigenfunctions) are actually unbounded \cite{BSimon3}. 
To address this issue, we restrict our attention to a bounded region which contains most of the mass at low temperatures.
Restricting to this region requires an additional logarithmic factor in the sample size, and a polynomially longer running time when compared to the results in~\cite{han2025polynomialcomplexitysamplingmultimodal}.

Another essential estimate in \cite{han2025polynomialcomplexitysamplingmultimodal}, specifically \cite[Lemma 4.8]{han2025polynomialcomplexitysamplingmultimodal}, also fails in this $\R^d$ setting. In  \cite[Lemma 4.8]{han2025polynomialcomplexitysamplingmultimodal}, the ratio of one single particle's probability density over the stationary distribution is bounded in $L^{\infty}$ norm, regardless of the particle's starting position. However, this bound only works in the compact setting and blows up to infinity if the domain is unbounded. Thus, the proof of Theorem~\ref{t:mainIntro} relies on a pointwise mixing estimate that allows us to quantify the mixing of particles based on their initial positions. In particular, we obtain a pointwise bound on the Langevin transition kernel, derived by modifying  Nash's estimate on the parabolic kernel functions \cite{nashPara}. We combine this bound on the transition kernel with a coupling argument inspired by \cite{Marion23SMC} to estimate the error in the high temperature regime, which is different from the spectrum decomposition argument in \cite{han2025polynomialcomplexitysamplingmultimodal}.

Finally, we make a few remarks on the assumptions and constants in Theorem~\ref{t:mainIntro}. The assumption that~$U$ has a double-well structure is mainly to simplify the technical presentation. The proof can be extended, with minor but tedious modifications, to the situation where~$U$ has more than two wells. Thus we only present the detailed proof of Theorem~\ref{t:mainIntro} in the double-well setting and provide a sketch of the proof in the multi-well case in Section~\ref{sec:multi}. The above assumption that the wells have equal depth above is also only for simplicity.
Our main result (Theorem~\ref{thm: main})  applies to potentials where each well contains a non-negligible portion of the total probability mass. In particular, it is shown in \cite[Lemma 4.4]{han2025polynomialcomplexitysamplingmultimodal} that if the wells have \emph{nearly equal depth}, then Theorem~\ref{thm: main} applies. We also remark that our result requires no prior knowledge of the location or the depth of the wells and only requires access to the energy and its gradient. The constants~$C_N$,~$C_T$ in Theorem~\ref{t:mainIntro} are independent of the temperature, but their dependence in dimension are not explicit since the proof involves bounding
the inner-product between the normalized eigenfunctions at successive temperature
levels, similar as in \cite{han2025polynomialcomplexitysamplingmultimodal}. 

\subsection*{Plan of the paper}

In Section~\ref{s:litReview}, we give a brief review of the literature. In Section~\ref{sec:main} we precisely state our algorithm, and state results guaranteeing convergence both for ASMC for a double-well energy function (Theorem~\ref{thm: main}, which generalizes Theorem~\ref{t:mainIntro}). We prove Theorem~\ref{thm: main} in Section~\ref{s:mainProof}. The remainder of this paper is devoted to the prove the required lemmas in the proof of Theorem~\ref{thm: main}.

\subsection*{Acknowledgements}
The authors would like to thank Gautam Iyer and Dejan Slep\v{c}ev, for helpful comments and discussions.

\subsection{Related work}\label{s:litReview}
The literature on sampling techniques is vast. Rather than attempting a comprehensive review, we refer the reader to~\cite{Chewi23, sanz2024book} for general overviews. We also refer the reader to the literature review session in our previous work \cite{han2025polynomialcomplexitysamplingmultimodal} for a more thorough overview on the related works. Here, we focus on several representative studies concerning sampling from multimodal distributions.

Sampling from multimodal distributions presents a well-known challenge. There is a broad spectrum of works studying algorithms that are suitable for sampling multimodal distributions. Among these, \emph{Sequential Monte Carlo (SMC)} and its variants, such as annealed importance sampling~\cite{Neal01}, population Monte Carlo~\cite{popMC}, and population annealing~\cite{popanneal}, have shown effective for multimodal targets. We refer to~\cite{DoucetFreitasEA01,ChopinPapaspiliopoulos20,SyedBouchardCoteEA24} for introduction to SMC methods. 

Several works rigorously analyze the non-asymptotic convergence  of SMC towards multimodal distributions,  under assumptions on the strong stability of the Markov kernels~\cite{schweizer2012SMC, Paulin19SMC}, or known partition of modes~\cite{Mathews24SMC,lee2024SMC}, or data-based initializations~\cite{koehler2024efficientlylearningsamplingmultimodal}. We refer to  the report~\cite{lee14provable} for a discussion on the provable guarantees under different assumptions of given information. Our previous work~\cite{han2025polynomialcomplexitysamplingmultimodal} relies on neither structural assumptions nor good initializations; however, it is restricted to compact domains.

Related MCMC-based approaches include parallel tempering~\cite{Swendsen1986tempering,geyer1991markov}  and simulated tempering~\cite{MarinariParisi92,Geyer1995}. Several notable works have rigorously analyzed the convergence of parallel and simulated tempering~\cite{WoodardSchmidlerEA09, WoodardSchmidlerEA09a,GeLeeEA18, GeLeeEA20,garg2025restrictedspectralgapdecomposition}. Further tempering and annealing methods include tempered transitions introduced by Neal in~\cite{Neal96a}, tempered Hamiltonian Monte Carlo~\cite{Nea11} and
the annealed Langevin Monte Carlo~\cite{guo2025provable,VacherChehabStromme25}. 

Recent work explores new strategies for multimodal sampling, including diffusion-model–inspired methods~\cite{VacherChehabStromme25}, adaptive MCMC~\cite{Pompe_20} and
ensemble-based methods~\cite{lu2019acceleratinglangevinsamplingbirthdeath,Lindsey22teleporting,Lu_2023}. Some approaches modify the Langevin dynamics to accelerate transitions between modes, either through altered diffusion~\cite{engquist2024samplingadaptivevariancemultimodal} or added drift terms~\cite{ReyBelletSpiliopoulos15,DamakFrankeEA20,ChristieFengEA23}.

\section{Main Result}\label{sec:main}
We begin by precisely stating the assumptions in the main theorem. The assumptions are very similar to that in~\cite{han2025polynomialcomplexitysamplingmultimodal} except that here the domain is $\R^d$. The first assumption requires~$U$ to be a regular,
double-well function with nondegenerate critical points. 

\begin{assumption}\label{a:criticalpts}
 The function~$U\in C^6(\mathbb R^d,\mathbb R)$, has a nondegenerate Hessian at all critical points, and has exactly two local minima located at~$x_{\min, 1}$ and~$x_{\min, 2}$. Moreover, there exist constants~$c, C_{U}>0$ such that 
    \begin{equation}\label{eq:Blap}
        |\Delta U|\leq c,\quad \text{and}\quad\liminf_{|x|\to\infty}\abs{\nabla U}\geq C_{U}.
    \end{equation}
  
\end{assumption}
We remark that important examples satisfying Assumption~\ref{a:criticalpts} include the Gaussian mixture in~$\R^d$ and the mixture of sub-exponential distribution with heavier tails than Gaussian. We normalize~$U$ so that
\begin{equation}\label{e:Upositive}
  0 = U(x_{\min, 1} )  \leq U( x_{\min, 2} )
  .
\end{equation}

Our next assumption concerns the saddle between the local minima~$x_{\min, 1}$ and~$x_{\min, 2}$.
Define the saddle height between~$x_{\min,1}$ and~$x_{\min,2}$ to be the minimum amount of energy needed to go from the global minimum~$x_{\min, 1}$ to~$x_{\min, 2}$, and is defined by
\begin{equation}\label{e:UHatDef}
    \hat{U} = \hat{U}(x_{\min,1},x_{\min,2}) \defeq \inf_\omega \sup\limits_{t\in [0,1]}U(\omega(t))
    .
\end{equation}
Here~$\omega \in C( [0, 1]; \mathbb R^d)$ are the paths such that~$\omega(0) =x_{\min,1}$,~$\omega(1) = x_{\min,2}$.

We also need to assume a nondegeneracy condition on the saddle.
\begin{assumption}\label{assumption: nondegeneracy}
   The saddle height between~$x_{\min,1}$ and~$x_{\min,2}$ is attained at a unique critical point~$s_{1,2}$ of index one.
   That is, 
   the first eigenvalue of~$\Hess U(s_{1,2})$ is negative and the others are positive.
\end{assumption}

We now introduce the notion of energy barrier and saddle height.
The \emph{energy barrier}, denoted by~$\hat{\gamma}$, is defined to be the minimum amount of energy needed to go from the minimum~$x_{\min, 2}$ to the global minimum~$x_{\min, 1}$.
In terms of~$s_{1, 2}$, the energy barrier~$\hat \gamma$ and the saddle height are given by
\begin{equation}\label{e:gammaHatDef}
    \hat{\gamma}\defeq U(s_{1,2})-U(x_{\min,2}),
    \quad\text{and}\quad
    \hat U = U(s_{1,2})
    .
\end{equation}

Then, we require the distribution~$\pi_\eta$ to be truly multimodal in the temperature range of interest.
That is, we require  sufficient mass for both \emph{basins of attraction} around~$x_{\min, 1}$ and~$x_{\min, 2}$.
The basin of attraction around~$x_{\min, i}$, denoted by~$\Omega_i$, is the set of all initial points for which the gradient flow of~$U$ eventually reaches~$x_{\min, i}$.
Precisely,~$\Omega_{i}$ is defined by
\begin{equation}
  \Omega_i
    \defeq
    \set[\Big]{
      y\in\mathbb R^d \st
      \lim_{t\to\infty}y_t=x_{\min,i}, \text{ where }\dot{y}_t=-\nabla U(y_t)
	\text{ with } y_0=y
      }
    ,
\end{equation}
and our multimodality condition is as follows.

\begin{assumption}\label{a:massRatioBound}
  There exists~$0 \leq \eta_{\min} < \eta_{\max} \leq \infty$, a constant~$C_m$ such that
   \begin{equation}\label{e:massRatioBound}
     \inf_{\substack{
       \epsilon\in [\eta_{\min},\eta_{\max}]\\
       0 < \epsilon < \infty
     }}
      \pi_{\epsilon}(\Omega_i)
	\geq \frac{1}{C_m^2}
	.
   \end{equation}
\end{assumption}

Finally, we require our initial points~$y_1^{1},\dots, y_1^{N}$ to start in some fixed region when~$N$ increases. For example, we can always start with points randomly distributed within a cube~$[-1,1]^d$, or start from delta mass on the origin. Precisely, we have the following assumption.
\begin{assumption}\label{assum:startgood}
    The constant
\begin{equation}
    C_{\mathrm{ini}}\equiv C_{\mathrm{ini}}(U)\defeq \max_{i=1,\dots,N}U(y_1^{i})
\end{equation}
is independent of~$N$.
\end{assumption}

We are now in a position to state our main result precisely.
\begin{theorem}\label{thm: main}
  Suppose for some~$0 \leq \eta_{\min} < \eta_{\max} \leq \infty$, the function~$U$ is a double-well function that satisfies 
  Assumptions~\ref{a:criticalpts},~\ref{assumption: nondegeneracy} and~\ref{a:massRatioBound} in Section~\ref{s:mainProof} below. Let~$\hat{\gamma}_r$ be the ratio of the saddle height~$\hat U$ to the energy barrier~$\hat \gamma$ given by $\hat \gamma_r\defeq \hat{U}/\hat{\gamma}$
and let
the process~$Y_{\epsilon , \cdot}$ in Algorithm~\ref{a:ASMC} given by~\eqref{e:Langevin}.
  Given~$\eta_1 \in (\eta_{\min}, \eta_{\max}]$ finite, for every~$\alpha, \delta, \nu > 0$ and~$\theta\in (0,1)$, there exist dimensional constants~$C_T = C_T(\alpha, \nu, C_{\mathrm{ini}}, U/\eta_1)$ and~$C_{N}(\alpha,\nu, U/\eta_1)$ such that, 
if choose~$M, N \in \N$, and~$T \in \R$ as \begin{align}\label{e:MTN}
    M &\geq \ceil*{ \frac{1}{\nu \eta} },
    \quad \text{and}\quad
    N \geq \frac{C_{N} M^2}{\delta^2}\log\paren[\Big]{\frac{M}{\theta}},\\
    T&\geq C_T\paren[\Big]{ \paren[\Big]{\frac{MN}{\theta}}^{\hat{\gamma}_r(1+\alpha)}\paren[\Big]{\log N+\log\paren[\Big]{\frac{M}{\theta}} }+\log\paren[\Big]{\frac{1}{\delta}}+\frac{1}{\eta} },
  \end{align}
   then for every bounded test function~$h$ and arbitrary initial points~$\set{y^i_1}$ satisfying Assumption~\ref{assum:startgood}, the points~$(x^1, \dots, x^N)$ returned by Algorithm~\ref{a:ASMC}  satisfy that \begin{equation}\label{e:MCerror}
     \P\paren[\Big]{ \abs[\Big]{
	\frac{1}{N}\sum_{i=1}^{N}h(x^i)-\int h(x)\pi_{\eta}(x)\,d x}
	  < \|h\|_{\osc} \delta }\geq 1-\theta
	.
  \end{equation}
\end{theorem}

We remark that Assumptions~\ref{a:criticalpts}--\ref{a:massRatioBound} are nondegeneracy assumptions, and do not require symmetry, or similarity of the shape of the wells. The proof of Theorem~\ref{thm: main} involves several technical lemmas controlling the pointwise bound of the Langevin transition kernel. It also uses the estimates in \cite{han2025polynomialcomplexitysamplingmultimodal} on the pointwise bound and shape of the eigenfunctions, which introduces dimensional pre-factors that are not explicit.
This takes up the rest of this paper and begins in Section~\ref{s:mainProof}, below.

\section{Proof of the Main Result}\label{s:mainProof}

\subsection{Notation and convention}\label{sec:notation}
Before showing the proof of Theorem~\ref{thm: main},  we
briefly list notational conventions that will be used throughout this paper. The notational conventions is the same as that of~\cite{han2025polynomialcomplexitysamplingmultimodal}, for completeness we include them here.
\begin{enumerate}[(i)]
  \item
    We will always assume~$C >0$ is a finite constant that can increase from line to line, provided it does not depend on the temperature~$\eta$.

  \item
    We use the convention that the expectation operator~$\E$ has lower precedence than multiplication.
    That is~$\E XY$ denotes the expectation of the product~$\E [XY]$, and~$\E X^2$, denotes the expectation of the square~$\E [X^2]$

  \item 
    When taking expectations and probabilities, a subscript will denote the conditional expectation / conditional probability.
    That is~$\E_X Y = \E(Y \given X)$ denotes the conditional expectation of~$Y$ given the~$\sigma$-algebra generated by~$X$.
  \item 
    When averaging functions of Markov processes, a superscript will denote the initial distribution.
    That is~$\E^\mu f(Y_t)$ denotes~$\E f(Y_t)$ given~$Y_0 \sim \mu$.
    When~$\mu = \delta_y$ is the Dirac~$\delta$-measure supported at~$y$, we will use~$\E^y$ to denote~$\E^{\delta_y}$.

  \item
    We interchangeably use~$\pi_\epsilon$ to denote the measure and the density.
    That is for~$x \in \R^d$,~$\pi_\epsilon(x)$ is given by~\eqref{e:piNu}, however for Borel sets~$A$,~$\pi_\epsilon(A)$ denotes~$\int_A \pi_\epsilon(x) \, dx$.
\end{enumerate}

\subsection{Proof of Theorem~\ref{thm: main}}

In this section, we prove Theorem~\ref{thm: main}. For simplicity and without loss of generality we assume~$\eta_1=1$.  Let~$\alpha, \theta,  \eta,\nu>0$ be fixed.

We begin by rewriting our algorithm in a manner that that is convenient for the proof.
Fix~$T > 0$ and~$N \in \N$ that will be chosen later.

\restartsteps
\step
  We start with~$N$ arbitrary points~$y^1_1$, \dots,~$y^N_1$ satisfying Assumption~\ref{assum:startgood}.

\step[Langevin step]
  For each~$k \in \set{1, \dots, M}$, and~$i \in \set{1, \dots N}$, let~$X^i_{k, \cdot}$ be the solution to the overdamped Langevin equation~\eqref{e:Langevin} with initial data~$X^i_{k, 0} = y^i_k$, driven by independent Brownian motions.

\step[Resampling step]
  Given the processes~$\set{X^i_{k, \cdot} \st i \leq N, k \leq M-1}$, we choose~$\set{y^1_{k+1}, \dots, y^N_{k+1}}$ independently from~$\set{X^1_{k, T}, \dots, X^N_{k, T}}$ so that
  \begin{equation}
    \P( y^i_{k+1} = X^j_{k, T} ) = \frac{\tilde r_k(X^j_{k, T})}{\sum_{i=1}^N \tilde r_k(X^i_{k,T})}
    .
  \end{equation}
  Here~$\tilde r_k$ is the ratio defined by~\eqref{e:rkTilde}.

\smallskip
We now briefly recall a few standard facts about the overdamped Langevin
dynamics~\eqref{e:Langevin} that will be used in the proof. For a more detailed introduction, we refer to~\cite[Section 4]{han2025polynomialcomplexitysamplingmultimodal}. Let~$L_{\epsilon}$ be the generator of~\eqref{e:Langevin}, whose action on smooth test functions is defined by
\begin{equation}\label{e:Ldef}
  L_{\epsilon} f \defeq-\epsilon\Delta f +\grad U\cdot\grad f
  .
\end{equation}
  Let~$L_\epsilon^*$ be the dual operator defined by
  \begin{equation}\label{e:LStarDef}
    L_\epsilon^* f = -\dv( \grad U f ) - \epsilon \lap f
    \,.
  \end{equation}
  It is well known~\cite[Chapter 8]{Oksendal03} that if~$Y_{\epsilon, \cdot}$ solves~\eqref{e:Langevin} starting from point~$x$, then for any~$t>0$, its transition density
  \begin{equation}\label{eq:defp}
      p_{\epsilon, t}(x,\cdot) \defeq \PDF(Y_{\epsilon, t})
  \end{equation}
  satisfies the \emph{Fokker-Planck equation}, a.k.a. the \emph{Kolmogorov forward equation}
  \begin{equation}\label{e:KForward}
    \partial_t p_{\epsilon} + L_\epsilon^* p_{\epsilon} = 0
    .
  \end{equation}
  One can readily check that the Gibbs distribution~$\pi_\epsilon$ is a stationary solution of~\eqref{e:KForward}, and hence must be the stationary distribution of~\eqref{e:Langevin}.
A direct calculation shows that
\begin{equation}\label{e:fOverPiNu}
  \partial_t \paren[\Big]{\frac{p_{\epsilon}}{\pi_\epsilon}}
    + L_\epsilon \paren[\Big]{\frac{p_{\epsilon}}{\pi_\epsilon}}
    = 0
    .
\end{equation}

The mixing properties of Langevin dynamics can be deduced directly from the spectral properties of the operator~$L_\epsilon$.
It is well known (see for instance~\cite[Chapter 8]{Kolokoltsov00}) that on the weighted space~$L^2(\pi_{\epsilon})$ the operator~$L_\epsilon$ is self-adjoint and has a discrete spectrum with eigenvalues
\begin{equation}
  0 = \lambda_{1, \epsilon}
    < \lambda_{2, \epsilon}
    \leq \lambda_{3, \epsilon}
    \cdots
\end{equation}
with corresponding~$L^2(\pi_\epsilon)$ normalized eigenfunctions~$\psi_{1, \epsilon}$,~$\psi_{2, \epsilon}$, etc. The first eigenvalue~$\lambda_{1, \epsilon} = 0$ corresponds to the constant eigenfunction~$\psi_{1, \epsilon} \equiv 1$. In our situation, because~$U$ has two wells, according to~\cite[Propositions~2.1, 2.2, Chapter~8]{Kolokoltsov00}, for every~$\gamma < \hat \gamma$ there exists constants~$\tilde{C}_\gamma$ and~$\Lambda$
(independent of~$\epsilon$)
such that
\begin{equation}\label{eq: egvalLEpsilon}
  \lambda_{2,\epsilon}\leq \tilde{C}_\gamma \exp\paren[\Big]{-\frac{\gamma}{\epsilon}}
  \quad\text{and}\quad
  \lambda_{i,\epsilon}\geq \Lambda,\quad \forall i\geq 3
  .
\end{equation}

The above facts, especially \eqref{e:fOverPiNu} and \eqref{eq: egvalLEpsilon}, will be used to show the fact that, if the particles are initialized within some bounded set~$K$, they are expected to mix locally after a relatively short time. Consequently, the Monte Carlo error in the Langevin step is primarily determined by the initialization error at each level, as stated in the following Lemma~\ref{l:langevinError}. Before precisely stating the lemma, we define the specific bounded set $K$ that will be used frequently in the subsequent context.

Define the subset~$K$ as
\begin{equation}\label{eq:defK}
    K\defeq B_1\cup B_2,
\end{equation}
where the sets~$B_i$,~$i\in\set{1,2}$ are defined by
\begin{equation}\label{eq:defB}
    B_i\defeq \set[\Big]{x\in\Omega_i  \st U(x)-U(x_{\min,i})\leq \frac{\hat \gamma}{(1+\alpha)^{\frac14}}}.
\end{equation}

It is straightforward to see that the subset~$K$ is bounded. We now state Lemma~\ref{l:langevinError}.
\begin{lemma}\label{l:langevinError}
  Assume that for each~$i \in \set{1, \dots, N}$,~$Y_{\epsilon,0}^{i}\in K$. Then  for any bounded test function~$h$, with probability larger than or equal to  \begin{equation}\label{e:LProb}
    1-2\exp\paren[\Big]{-\frac{2Na^2}{\norm{h}^2_{\osc}}}
\end{equation}
  we have
  \begin{align}
    \abs[\Big]{ \frac{1}{N}\sum_{i=1}^{N}h(Y^{i}_{\epsilon, T})-\int h\pi_{\epsilon}\,d x}
    &\leq e^{-\lambda_{2, \epsilon} T} \abs[\Big]{\int h\psi_{2,\epsilon}\pi_{\epsilon} \, dx} \cdot\abs[\Big]{\frac{1}{N}\sum_{i=1}^{N}\psi_{2,\epsilon}\one_{K}(Y^{i}_{\epsilon,0})}
    \\
    &
    + \mathcal E_{\epsilon,T}(h)  +a\label{e:langevinError}
\end{align}
where
\begin{equation}\label{eq: def of tilde epsilon}
  \mathcal E_{\epsilon,T}(h)\defeq \norm{h}_{L^2(\pi_{\epsilon})} e^{-\Lambda T/2}\max_{x\in K}\sqrt{\frac{p_{\epsilon,T}(x,x)}{\pi_{\epsilon}(x)}-1}.
\end{equation}
\end{lemma}  Lemma~\ref{l:langevinError} may look similar to \cite[Lemma 4.6]{han2025polynomialcomplexitysamplingmultimodal}, they actually differ in several aspects. As we mentioned earlier, the eigenfunctions~$\psi_{2,\epsilon}$ are no longer bounded and thus cannot serve as test functions. To mitigate this, we restrict~$\psi_{2,\epsilon}$ on bounded set~$K$ and later compensate this truncation in some bias terms.  Moreover, we no longer have control on the error of marginal distribution which appears in the higher order term in \cite[Lemma 4.6]{han2025polynomialcomplexitysamplingmultimodal}. Instead, we derive a pointwise bound on the transition density to estimate the higher order term \eqref{eq: def of tilde epsilon}, which is precisely stated in the following lemma. 
\begin{lemma}\label{cor:poverpi}
    For~$x\in\mathbb{R}^d$,~$\epsilon<1$, there exists constant~$C_p\equiv C_p(d,U)$ such that
    \begin{equation}\label{eq:poverpi}
       \frac{p_{\epsilon,t}(x,x)}{\pi_{\epsilon}(x)}\leq C_{p}\exp\paren[\Big]{\frac{U(x)}{\epsilon}}\paren[\Big]{1-\exp\paren[\Big]{-\frac{c}{d}t}}^{-\frac{d}{2}}.
    \end{equation}
    Consequently, 
    \begin{equation}\label{eq:poverpiK}
       \max_{x\in K}\frac{p_{\epsilon,t}(x,x)}{\pi_{\epsilon}(x)}\leq C_p\exp\paren[\Big]{\frac{\hat{U}}{\epsilon}}\paren[\Big]{1-\exp\paren[\Big]{-\frac{c}{d}t}}^{-\frac{d}{2}}.
    \end{equation}
\end{lemma}
The proof of Lemma \ref{cor:poverpi} is based on a modification of Nash's estimate on the pointwise bound of the kernel function of parabolic equations \cite{nashPara}, where the original result concerns general divergence-form operators and the dependence on the coefficients is implicit. Additional work is required to derive the explicit~$\epsilon$-dependence.

Observe that the term~$\psi_{2,\epsilon}\one_K$ on the right hand side of \eqref{e:langevinError} indicates the major reason for choosing~$K$ as in \eqref{eq:defK}: we can actually obtain a pointwise bound on~$\psi_{2,\epsilon}\one_K$ that is uniform in~$\epsilon$ for~$\epsilon\leq 1$. Precisely,
\begin{lemma}\label{lem:boundPsi}
     There exists constant~$C_{\psi} = C_{\psi}(U,d,C_m,\alpha)>0$ 
    such that
    \begin{equation}\label{eq:defCpsi}
         \sup\limits_{0< \epsilon\leq 1}\|\psi_{2,\epsilon}\one_{K}\|_{L^{\infty}}\leq C_{\psi}.
    \end{equation}
\end{lemma}

\smallskip

The above lemmas give the estimates for the error in the Langevin step, we now turn to the resampling step. Here we  have the following Lemma~\ref{l:rebalancing} estimating the Monte Carlo error, which is very similar to \cite[Lemma 3.3]{han2025polynomialcomplexitysamplingmultimodal}. The difference is that \cite[Lemma 3.3]{han2025polynomialcomplexitysamplingmultimodal} estimates the variance of the Monte Carlo error and here it controls the error with high probability. Due to this difference, we state the lemma below and present the proof in Section \ref{s:resample}.

\begin{lemma}\label{l:rebalancing}
 Suppose~$x^1$, \dots,~$x^N$ are~$N$ (not necessarily i.i.d.) random points in~$\mathcal X$.
  Let~$\tilde p$,~$\tilde q \colon \mathcal X \to [0, \infty)$ be two unnormalized probability density functions, and choose~$y^1,\dots,y^N$ independently from~$\set{x^1, \dots, x^N}$ according to
  \begin{equation}\label{e:PyiEqXj}
  \P( y^i = x^j ) = \frac{\tilde r(x^j)}{\sum_{i=1}^N \tilde r(x^i)},
  \quad\text{where}\quad
  \tilde r \defeq \frac{\tilde q}{\tilde p}
  .
\end{equation}
Let~$p, q$ are the normalized probability distributions corresponding to~$p, q$ respectively.  Then for any bounded function~$h$,  with probability larger than or equal to 
  \begin{equation}\label{e:probResam}
     1- 2\exp\paren[\Big]{-\frac{2Na^2}{\norm{h}^2_{\osc}}},
  \end{equation}
  we have
  \begin{align}
    \abs[\Big]{\frac{1}{N} \sum_1^N h(y^{i}) - \int_{\mathcal{X}} hq\, dx} &\leq  \norm[\Big]{h-\int_{\mathcal X} h q\, dx}_{L^\infty} \abs[\Big]{1- \frac{1}{N}\sum_{i=1}^N r(x^i)}\\
     &+\abs[\Big]{\frac{1}{N}\sum_{i=1}^N r(x^i)\Big(h(x^i)  -\int_{\mathcal X} h q \, dx \Big)}+a. \label{e:resampling}
  \end{align}
  Here~$r$ is the ratio
  \begin{equation}\label{e:rdef}
    r \defeq \frac{q}{p}
    ,
    \quad\text{where}\quad
    p = \frac{\tilde p}{\int_{\mathcal X} \tilde p \, dx}
    \quad\text{and}\quad
    q = \frac{\tilde q}{\int_{\mathcal X} \tilde q \, dx}
    .
  \end{equation}
  \end{lemma}

\smallskip

We now explain the main idea behind the proof of Theorem~\ref{thm: main}. We split our analysis into the high and low temperature regime, respectively. When temperature is relatively high, we let the Langevin to mix globally and when the temperature is low, we only require local mixing. It is thus essential to choose a proper threshold for the division of low and high temperature. This critical temperature~$\eta_{\mathrm{cr}}$ depends on number of levels~$M$ and sample size~$N$, which will be determined later. We remark that the choice of the critical temperature is used solely in the proof and does not appear in the implementation of the algorithm.

Given the critical temperature~$\eta_{\mathrm{cr}}$, we define the critical level as
\begin{equation}\label{eq:defkcr}
    k_{\mathrm{cr}}\defeq\min\Bigl\{2\leq k\leq M \: | \: \eta_k\leq \eta_{\mathrm{cr}}\Bigr\}.
\end{equation}
We split our analysis into the high temperature regime~$k< k_{\mathrm{cr}}$ and the low temperature regime~$k\geq k_{\mathrm{cr}}$. We first consider the  low temperature regime and then the high temperature regime.

\smallskip
\noindent\emph{Low temperature regime.} 
In the low temperature regime, Lemma~\ref{l:langevinError} and Lemma~\ref{l:rebalancing} allows us to derive Monte Carlo error estimates between levels~$k$ and~$k+1$ in Algorithm~\ref{a:ASMC}.
Recall in Algorithm~\ref{a:ASMC},~$M$ is chosen according to~\eqref{e:MTN},~$\eta_1 = 1$,~$\eta_M = \eta$, and~$1/\eta_1$, \dots,~$1/\eta_M$ are linearly spaced.
That is,~$\eta_k$ is chosen according to
\begin{equation}\label{e:chooseEtaK}
  \eta_k
    \defeq \frac{(M - 1) \eta}{(M-1) \eta + (k-1)(1 - \eta) }
    .
\end{equation}

For~$k \in \set{1, \dots, M}$, by a slight abuse of notation we define
\begin{equation}\label{e:piKDef}
  \pi_k \defeq \pi_{\eta_k} 
  ,
  \quad
  \tilde \pi_k \defeq \tilde \pi_{\eta_k}
  \quad\text{and}\quad
  Z_k \defeq Z_{\eta_k}
\end{equation}
where~$\pi_{\eta_k}$,~$\tilde \pi_{\eta_k}$ and~$Z_{\eta_k}$ are defined by~\eqref{e:piNu} with~$\epsilon = \eta_k$.
Next we define
\begin{equation}\label{e:defrk}
  r_k \defeq \frac{\pi_{k+1}}{\pi_k}
\end{equation}
to be the ratio of \emph{normalized} densities at levels~$k+1$ and~$k$.
In practice, we do not have access to~$r_k$ as we do not have access to the normalization constants~$Z_k$.
This is why Algorithm~\ref{a:ASMC} is formulated using the ratio of unnormalized densities~$\tilde r_k$ defined in~\eqref{e:rkTilde}.
It is shown in \cite[Lemma 9.2]{han2025polynomialcomplexitysamplingmultimodal} that there exists constant~$C_r\equiv C_r(U,\nu)$ such that 
\begin{equation}\label{eq:defCr}
    C_r \defeq \max_{1\leq k\leq M-1}\|r_k\|_{L^\infty(\R^d)}
    .
\end{equation}
Moreover,~$C_r>1$ and~$C_r\to 1$ as~$\nu\to 0$.

For simplicity of notation, we use a subscript of~$k$ on the error, eigenvalue and eigenfunction to denote the corresponding quantities at~$\epsilon = \eta_k$.
Explicitly, we write
\begin{equation}
    \lambda_{2,k}\defeq \lambda_{2,\eta_k}
    , \quad
    \psi_{2,k}\defeq \psi_{2,\eta_k}.
\end{equation}

We now state the key estimate used to carry out the analysis in the low-temperature regime, whose proof is given in Section~\ref{sec:iteration}.
\begin{lemma}\label{l:iteration}
  Choose~$M$ as in~\eqref{e:MTN} and~$\eta_k$ as in~\eqref{e:chooseEtaK}. Fix~$\theta\in (0,1)$.
  Define constant~$C_K\equiv C_K(\alpha,\hat{\gamma})$ as
\begin{equation}\label{eq:defCK}
    C_K\defeq \frac{\hat{\gamma}}{(1+\alpha)^{\frac12}}.
\end{equation}
  There exist dimensional constants~$C_{\alpha} = C_{\alpha}(\alpha,U)>0$ and~$C_{\mathrm{tem}}= C_{\mathrm{tem}}(\alpha,U)>0$ such that for any~$\delta>0$, if
  \begin{align}
 \label{eq:Nchoice} N&\geq 128C^2_r(C_{\psi}+1)^2 \frac{M^2}{\delta^2}\log\paren[\Big]{\frac{8M}{\theta}},\\
 \label{e:TNlow}   T&\geq \max\Bigl\{C_{\alpha}\Big(\log\paren[\Big]{\frac{M}{\theta}}+\log N+\log\paren[\Big]{\frac{1}{\delta}}+\frac{1}{\eta}\Big), \frac{c}{d}\log (4d)\Bigr\},
  \end{align}
  and let
  \begin{equation}\label{e:CriEtaLT}
      \eta_{\mathrm{cr}}=\frac{C_K}{\log\paren[\Big]{\frac{C_{\mathrm{tem}}MN}{\theta}}},
  \end{equation}
  then for each~$ k_{\mathrm{cr}}\leq k\leq M-1$, 
  with probability 
  \begin{equation}\label{eq:thetakLB}
      \theta_k\geq 1-\frac{\theta}{M},
  \end{equation}
  we have
\begin{equation}\label{eq:Xk+1inK}
        X_{k+1,0}^{i}\in K,\quad\forall i
    \end{equation}
    and
     \begin{equation}\label{e:iteration}
     \abs[\Big]{\frac{1}{N}\sum_{i=1}^{N}\psi_{2,k+1}\one_{K}(X_{k+1,0}^i)}\leq \beta_k \abs[\Big]{\frac{1}{N}\sum_{i=1}^{N}\psi_{2,k}\one_{K}(X_{k,0}^i)}+c_k. 
  \end{equation}
  Here the constants~$\beta_k$,~$c_k$ are such that for every~$k$ such that~$\eta_k$ satisfies \eqref{e:CriEtaLT} we have
  \begin{equation}\label{e:CBeta}
    \prod_{j=k}^{M-1} \beta_j \leq C_{\beta}
    \quad\text{and}\quad
    c_k \leq \frac{\delta}{M}
    ,
  \end{equation}
  for some dimensional constant~$C_{\beta}>1$ (independent of~$\alpha$,~$\delta$).
\end{lemma}
The intuition behind \eqref{eq:thetakLB} and \eqref{eq:Xk+1inK} is that, when the temperature is low, it is expected that the~$\pi_k$ has tiny mass outside~$K$; thus, it is also a rare event that the particle falls outside~$K$. The proof inequality \eqref{e:iteration} and \eqref{e:CBeta} are similar to that of \cite[Lemma 4.7]{han2025polynomialcomplexitysamplingmultimodal}, which are based on a careful estimates on the inner product of eigenfunctions on subsequent levels. Moreover, Lemma~\ref{l:iteration} foreshadows that the critical temperature mentioned earlier will be chosen as in \eqref{e:CriEtaLT}.

\smallskip
\noindent\emph{High temperature regime.}
In the high-temperature regime, the global mixing time is reasonably short. We have the below Lemma \ref{lem:globalmixing}, whose proof is in Section~\ref{sec:global}  following the idea of \cite{Marion23SMC}. In~\cite{Marion23SMC}, a similar result as Lemma \ref{lem:globalmixing} is established assuming that at the first level the samples are i.i.d following~$\pi_1$, which is not our case here. Thus we need to further use a maximal coupling argument for the estimate of mixing at the first level.
\begin{lemma}[global mixing] \label{lem:globalmixing}
 Fix~$p\in (0,\frac{1}{4})$,~$\mathfrak{e}>0$ and let~$1\leq k\leq k_{\mathrm{cr}}$. There exists constant~$\tilde{C}_{\alpha}\equiv \tilde{C}_{\alpha}(\alpha, U, C_{\mathrm{ini}})$ such that if 
\begin{align}
   N&\geq \frac{9C_r^2}{\mathfrak{e}^2}\log\paren[\Big]{\frac{16M}{p^2}},\\
   T&\geq \max\Bigl\{\tilde{C}_{\alpha}\exp\paren[\Big]{\frac{(1+\alpha)^{\frac12}\hat{U}}{\eta_{\mathrm{cr}}}}\paren[\Big]{\log(N)+\log\paren[\Big]{\frac{M}{p}} },
\frac{c}{d}\log (4d)\Bigl\}, \label{eq:NTglobal}
    \end{align}
   then for any continuous function~$f$ with~$|f|\leq 1$,
    \begin{equation}
        \P\paren[\Big]{\abs[\Big]{\frac{1}{N}\sum_{i=1}^{N}f(X_{k,T}^{i})-\int f\pi_{k}\, d x}<\mathfrak{e}}\geq 1-p.
    \end{equation}
\end{lemma}
Notice that if~$\eta_{\mathrm{cr}}$ is chosen as in \eqref{e:CriEtaLT} and~$N$ is chosen as in \eqref{eq:Nchoice}, then the time~$T$ in \eqref{eq:NTglobal} is actually polynomial in~$M$. In particular, when~$k=k_{\mathrm{cr}}$, Lemma~\ref{lem:globalmixing} implies the following result.
\begin{lemma} \label{lem:baseCasePsi} Fix~$\delta>0, \theta\in (0,\frac13)$. Choose~$N$ as in \eqref{eq:Nchoice} and let~$\eta_{\mathrm{cr}}$ be defined as in \eqref{e:CriEtaLT}. Let~$T$ satisfy \eqref{eq:NTglobal}. Then
\begin{equation}
    \P\paren[\Big]{\abs[\Big]{\frac{1}{N}\sum_{i=1}^{N}\psi_{2,k_{\mathrm{cr}}}\one_{K}(X^{i}_{k_{\mathrm{cr}},0})}< (C_r+1)\delta}\geq 1-3\theta.
\end{equation}
   
\end{lemma}

\smallskip
Now we are well-equipped to prove Theorem~\ref{thm: main}.

\begin{proof}[Proof of Theorem~\ref{thm: main}]
Fix~$\alpha>0, \theta\in (0,1), \delta>0$, and define
  \begin{equation}\label{e:TildeDelta}
      \tilde{\delta}=
    \frac{\delta}{C_{\beta}(C_r+2)},\quad \tilde{\theta}=\frac{\theta}{5}
  \end{equation}
where~$C_{\beta}$ is the constant in \eqref{eq: defCbeta} which depends on the given~$\alpha>0$. Choose~$M$ as in \eqref{e:MTN}.  
Then we define 
  \begin{align}
 \label{eq: parametersN}
  N&\geq 128C^2_r(C_{\psi}+1)^2 \frac{M^2}{\tilde\delta^2}\log\paren[\Big]{\frac{8M}{\tilde\theta}}\\
    T&\geq \max\Bigl\{C_{\alpha}\Big(\log\paren[\Big]{\frac{M}{\tilde{\theta}}}+\log N+\log\paren[\Big]{\frac{1}{\tilde\delta}}+\frac{1}{\eta}\Big),\\
    &\qquad\tilde{C}_{\alpha}\paren[\Big]{\frac{C_{\mathrm{tem}}MN}{\tilde{\theta}}}^{\hat{\gamma}(1+\alpha)}\paren[\Big]{\log N+\log\paren[\Big]{\frac{M}{\tilde{\theta}}} },\\
    &\qquad\frac{2}{\Lambda}\Big(\log\paren[\Big]{\frac{1}{\tilde{\delta}}}+ \frac{\hat{U}}{2\eta}+\frac{1}{8}+\log(C_p^{\frac12})\Big),
\frac{c}{d}\log (4d)\Bigl\}.   \label{eq: parametersT}
    \end{align}
Here~$C_{\alpha}$,~$\tilde{C}_{\alpha}$,~$C_{\mathrm{tem}}$ and~$C_p$ are the constants defined in \eqref{e:TNlow}, \eqref{eq:NTglobal}, \eqref{e:CriEtaLT} and \eqref{eq:poverpiK}, respectively.  Notice that if~$N, T$ are chosen according to~\eqref{eq: parametersN} and~\eqref{eq: parametersT}, then we can find constants~$C_{T} = C_{T}(\alpha, U, C_{\mathrm{ini}})>0$ and~$C_{N} = C_{N}(\alpha, U)>0$ so that this choice is consistent with the choice in~\eqref{e:MTN}.
  We will now show that~\eqref{e:MCerror} holds for any bounded test function~$h \in L^\infty(\R^d)$. Observe that the left hand side of \eqref{e:MCerror} does not change if we add a constant to~$h$. Thus without loss of generality we may replace~$h$ with~$h - \inf h + \frac{1}{2} \norm{h}_\osc$, and assume~
  \begin{equation}\label{e:h}
      \norm{h}_{L^\infty} = \frac{1}{2} \norm{h}_\osc.
  \end{equation}

Define the critical temperature~$\eta_{\mathrm{cr}}$ as in \eqref{e:CriEtaLT}.
 Notice that~$1/\eta_{\mathrm{cr}}$ is~$O(\log(M))$ while~$1/\eta$ is~$O(M)$, therefore when~$\eta$ is sufficiently small, we should have~$\eta_{\mathrm{cr}}<\eta$. We discuss different cases of~$\eta$. 

\restartcases

\case[$\eta\geq \eta_{\mathrm{cr}}$] This implies that~$k_{\mathrm{cr}}=M$. We remark that this might happen when~$M$ is relatively small. 
Observe that~$N$ and~$T$ satisfies \eqref{eq:NTglobal} with
\begin{equation}
 \mathfrak{e}=\delta, \quad p=\theta.
\end{equation}  
The fact that~$T$ satisfies \eqref{eq:NTglobal} is straightforward  from the observation that~$\tilde{\delta}<\delta$ and~$\tilde{\theta}<\theta$. We check for~$N$, notice that
\begin{align*}
    N&\geq 128C^2_r(C_{\psi}+1)^2 \frac{M^2}{\delta^2}\log\paren[\Big]{\frac{8M}{\theta}}=64C^2_r(C_{\psi}+1)^2 \frac{M^2}{\delta^2}\log\paren[\Big]{\frac{64M^2}{\theta^2}}\\
    &\geq \frac{9C_r^2}{\mathfrak{e}^2}\log\paren[\Big]{\frac{16M}{p^2}}.
\end{align*}
Applying Lemma~\ref{lem:globalmixing} to 
\begin{equation}
    f=h/\norm{h}_{L^{\infty}},\quad \mathfrak{e}=\delta, \quad \quad p=\theta
\end{equation}
 yields
 \begin{equation}
        \P\paren[\Big]{\abs[\Big]{\frac{1}{N}\sum_{i=1}^{N}h(X_{M,T}^{i})-\int h\pi_{M}\, d x}<\norm{h}_{L^{\infty}}\mathfrak{e}}\geq 1-\theta.
    \end{equation}
The proof completes by noticing \eqref{e:h}.

\case[$\eta<\eta_{\mathrm{cr}}$] Using Lemma~\ref{l:langevinError}, we obtain
  with probability larger than or equal to  \eqref{e:LProb},
  we have
  \begin{align}
    \MoveEqLeft\abs[\Big]{ \frac{1}{N}\sum_{i=1}^{N}h(X^{i}_{M, T})-\int h\pi_{M}\,d x}\\
    &\leq e^{-\lambda_{2, M} T} \abs[\Big]{\int h\psi_{2,M}\pi_{\epsilon} \, dx} \cdot\abs[\Big]{\frac{1}{N}\sum_{i=1}^{N}\psi_{2,M}\one_{K}(X^{i}_{M,0})}
    + \mathcal E_{M,T}(h) +a\\
    &=E_1+ \mathcal E_{M,T}(h) +a. \label{e:ErrMT1}
\end{align}
where
\begin{equation}\label{eq:defE1}
   E_1\defeq  e^{-\lambda_{2, M} T} \abs[\Big]{\int h\psi_{2,M}\pi_{\epsilon} \, dx} \cdot\abs[\Big]{\frac{1}{N}\sum_{i=1}^{N}\psi_{2,M}\one_{K}(X^{i}_{M,0})}.
\end{equation}
  We will now show that the right hand side of~\eqref{e:ErrMT1} is bounded above by~$\delta \norm{h}_\osc$.

 \restartsteps
\step[Choose~$a$] Observe that if we choose
\begin{equation}\label{e:achoiceh}
    a = \frac12\tilde{\delta} \norm{h}_\osc,
\end{equation}
then the choice of~$N$ \eqref{eq: parametersN} gives that 
\begin{equation}
    \P\paren[\Big]{\abs[\Big]{ \frac{1}{N}\sum_{i=1}^{N}h(X^{i}_{M, T})-\int h\pi_{M}\,d x}\leq E_1+ \mathcal E_{M,T}(h) +\frac12\tilde{\delta} \norm{h}_\osc }\overset{\eqref{e:LProb},\eqref{e:ErrMT1}}{\geq} 1-\frac{\tilde{\theta}}{M}. 
\end{equation}

 \step[Bound~$\mathcal E_{M,T}(h)$] 
 First notice that when
\begin{equation}\label{eq:TCrip}
    T\geq \frac{c}{d}\log (4d),
\end{equation}
we have that
\begin{equation}\label{eq:ubPatTcri}
    \paren[\Big]{1-\exp\paren[\Big]{-\frac{c}{d}T}}^{-\frac{d}{2}}\leq  \paren[\Big]{1+2\exp\paren[\Big]{-\frac{c}{d}T}}^{\frac{d}{2}}\leq \paren[\Big]{1+\frac{1}{2d}}^{\frac{d}{2}}\leq e^{\frac{1}{4}}.
\end{equation}
When~$T$ satisfies \eqref{eq: parametersT}, we have
\begin{equation}
    e^{-\Lambda T/2}\leq \exp\paren[\Big]{\frac{\hat{U}}{2\eta}}  e^{-\frac{1}{8}}\frac{\tilde{\delta}}{C_p^{\frac12}}
\end{equation}
which implies that
\begin{align}
    \mathcal E_{M,T}(h) &\overset{\eqref{eq: def of tilde epsilon}}{\leq} \norm{h}_{L^{\infty}} e^{-\Lambda T/2}\max_{x\in K}\sqrt{\frac{p_{k,T}(x,x)}{\pi(x)}-1}\\
    &\overset{\mathclap{\eqref{eq:poverpiK},\eqref{e:h},\eqref{eq:ubPatTcri}}}{\leq} \qquad\quad\frac{1}{2}\norm{h}_{\osc}\tilde{\delta}. \label{e:calEh}
\end{align}
 \step[Bound~$E_1$]
Notice that
\begin{align}
    E_1&\overset{\mathclap{\eqref{e:defE1}}}{\leq} \:\norm{h}_{L^{\infty}} \abs[\Big]{\frac{1}{N}\sum_{i=1}^{N}\psi_{2,M}\one_{K}(X^{i}_{M,0})}\\
    &\overset{\mathclap{\eqref{e:h}}}{\leq} \:\frac{1}{2}\norm{h}_{\osc}\abs[\Big]{\frac{1}{N}\sum_{i=1}^{N}\psi_{2,M}\one_{K}(X^{i}_{M,0})}, \label{e:defE1}
\end{align}
thus it suffices to bound~$\abs[\Big]{\frac{1}{N}\sum_{i=1}^{N}\psi_{2,M}\one_{K}(X^{i}_{M,0})}$.
 
Observe the choice of~$N$ and~$T$ in \eqref{eq: parametersN} and \eqref{eq: parametersT} satisfies \eqref{e:TNlow} with parameter $\tilde{\delta}$ and $\tilde{\theta}$. Using Lemma~\ref{l:iteration} repeatedly, we obtain that with probability larger than or equal to
\begin{equation}
   \prod_{k=k_{\mathrm{cr}}}^{M-1}\theta_k \geq \paren[\Big]{1 -\frac{\tilde{\theta}}{M}}^{M-k_{\mathrm{cr}}}\geq  \paren[\Big]{1 -\frac{\tilde{\theta}}{M}}^{M}\geq 1-\tilde{\theta},
\end{equation}
we have
\begin{align}\label{eq: est E psiM}
 \MoveEqLeft~\abs[\Big]{\frac{1}{N}\sum_{i=1}^{N}\psi_{2,M}\one_{K}(X^{i}_{M,0})} \\
 &\leq \paren[\Big]{\prod_{j=k_{\mathrm{cr}}}^{M-1}\beta_j}\abs[\Big]{\frac{1}{N}\sum_{i=1}^{N}\psi_{2,k_{\mathrm{cr}}}\one_{K}(X^{i}_{k_{\mathrm{cr}},0})}
       +\sum_{k=2}^{M-2}c_k \Big(\prod_{j=k+1}^{M-1}\beta_j\Big)+c_{M-1}\\
&\overset{\mathclap{\eqref{e:CBeta}}}{\leq} C_{\beta}\abs[\Big]{\frac{1}{N}\sum_{i=1}^{N}\psi_{2,k_{\mathrm{cr}}}\one_{K}(X^{i}_{k_{\mathrm{cr}},0})}
       +C_{\beta}\sum_{k=2}^{M-2}\frac{\tilde{\delta}}{M} +\frac{\tilde{\delta}}{M}\\
&\leq C_{\beta}\abs[\Big]{\frac{1}{N}\sum_{i=1}^{N}\psi_{2,k_{\mathrm{cr}}}\one_{K}(X^{i}_{k_{\mathrm{cr}},0})}+ C_{\beta}\tilde{\delta}.
\end{align}

According to Lemma~\ref{lem:baseCasePsi}, 
\begin{equation}\label{eq:ThmPsicr}
   \P\paren[\Big]{\abs[\Big]{\frac{1}{N}\sum_{i=1}^{N}\psi_{2,k_{\mathrm{cr}}}\one_{K}(X^{i}_{k_{\mathrm{cr}},0})}\leq (C_r+1)\tilde{\delta}}\geq 1-3\tilde{\theta}.
\end{equation}

Therefore, with probability larger than or equal to 
\begin{equation}
    (1-\tilde{\theta})\cdot 1-3\tilde{\theta}\geq 1-4\tilde{\theta},
\end{equation}
we have
\begin{equation}
    \abs[\Big]{\frac{1}{N}\sum_{i=1}^{N}\psi_{2,M}\one_{K}(X^{i}_{M,0})}\overset{\eqref{eq: est E psiM},~\eqref{eq:ThmPsicr}}{\leq} C_{\beta}(C_r+1)\tilde{\delta}+C_{\beta}\tilde{\delta}=C_{\beta}(C_r+2)\tilde{\delta}.
\end{equation}
which in turn gives that
\begin{equation}\label{e:estE1}
    E_1\overset{\eqref{e:defE1}}{\leq} \frac{1}{2}\norm{h}_{\osc}C_{\beta}(C_r+2)\tilde{\delta}.
\end{equation}

Using~\eqref{e:achoiceh},~\eqref{e:calEh} and~\eqref{e:estE1} in~\eqref{e:ErrMT1} implies that with probability larger than or equal to
\begin{equation}
   (1-\frac{\tilde{\theta}}{M}) \cdot(1-4\tilde{\theta})\geq 1-5\tilde{\theta}\overset{\eqref{e:TildeDelta}}{=} 1-\theta
\end{equation}
 we have
\begin{align}
   \MoveEqLeft\abs[\Big]{ \frac{1}{N}\sum_{i=1}^{N}h(X^{i}_{M, T})-\int h\pi_{M}\,d x}~~\overset{\mathclap{\eqref{e:ErrMT1}}}{\leq}~~E_1+ \mathcal E_{M,T}(h) +a\\
   &\overset{\mathclap{\eqref{e:achoiceh},\eqref{e:calEh},\eqref{e:estE1}}}{\leq}\qquad\quad\frac{1}{2}\norm{h}_{\osc}C_{\beta}(C_r+2)\tilde{\delta}+\frac{1}{2}\norm{h}_{\osc}\tilde{\delta}+\frac{1}{2}\norm{h}_{\osc}\tilde{\delta}
   \overset{\eqref{e:TildeDelta}}{\leq} \norm{h}_{\osc}\delta.
\end{align}
This proves~\eqref{e:MCerror}, concluding the proof.
\end{proof}

It remains to prove the Lemmas~\ref{l:langevinError},~\ref{cor:poverpi},~\ref{lem:boundPsi},~\ref{l:rebalancing},~\ref{l:iteration},~\ref{lem:globalmixing} and~\ref{lem:baseCasePsi}, which will be done in subsequent
sections.

\subsection{Sketch of proof in the multi-well case}\label{sec:multi}
In this section, we present the version of Theorem~\ref{thm: main} when the energy function~$U$ has more than two local minima and provide a detailed sketch of the proof.

Assume the potential~$U$ now has $J$ local minima located at~$x_{\min, 1},\dots, x_{\min, J}$. We need a modified nondegeneracy assumption of Assumption~\ref{assumption: nondegeneracy}.  This assumption is analogous to \cite[Assumption 1.7]{MenzSchlichting14}, which guarantees that the estimate \eqref{eq:lambdaLowBound} still holds in the multi-modes case. For completeness we state the assumption below.
\begin{assumption}[Assumption 1.7 in~\cite{MenzSchlichting14}]\label{assumption: nondegeneracyM}
   There exists $\delta>0$ such that:

  (\romannumeral1) The saddle height between two local minima $x_{\min,i}$ and $x_{\min,j}$ is attained at a unique critical point $s_{i,j}$ of index one.
   That is, 
   the first eigenvalue of $\Hess U(s_{1,2})$ is negative and the others are positive. The point $s_{i,j}$ is called communicating saddle between the minima $x_{\min,i}$ and $x_{\min,j}$.

(\romannumeral2) The set of local minima~$\set{x_{\min,1},\dots, x_{\min,J}}$ is ordered such that $x_{\min,1}$
is a global minimum and for all $i\in\set{3,\dots,J}$ yields
\begin{equation*}
    U(s_{1,2})- U(x_{\min,2}) \geq U(s_{1,i})- U(x_{\min,i}) +\delta.
\end{equation*}
\end{assumption}

In the multi-modes case, the \emph{energy barrier}~$\hat{\gamma}$ and the saddle height are given by
\begin{equation}\label{e:gammaHatDefM}
    \hat{\gamma}\defeq \min\Delta_i, \quad \Delta_i =\min_{x\in\partial\Omega_i}U(x)- U(x_{\min,i})
    ,\quad\text{and}\quad
    \hat U = U(s_{1,2}).
\end{equation}
The ratio $\hat{\gamma}_r$ is given by
\begin{equation}\label{e:gammaHatRDefM}
    \hat \gamma_r\defeq \frac{\hat{U}}{\hat{\gamma}}
    .
\end{equation}

Next, we state the multi-modes version of Theorem~\ref{thm: main}.
\begin{theorem}[Multi-modes]\label{thm: mainM}
  Suppose for some~$0 \leq \eta_{\min} < \eta_{\max} \leq \infty$, the function~$U$ is a function that satisfies 
Assumptions~\ref{a:criticalpts},~\ref{a:massRatioBound} and~\ref{assumption: nondegeneracyM}.
  Let~$\hat \gamma_r \geq 1$ be defined as in~\eqref{e:gammaHatRDefM}. In the same setting as in Theorem~\ref{thm: main}, for every bounded test function~$h$ and arbitrary initial points~$\set{y^i_1}$ satisfying Assumption~\ref{assum:startgood}, the points~$(x^1, \dots, x^N)$ returned by Algorithm~\ref{a:ASMC}  satisfy that~\eqref{e:MCerror}.
\end{theorem}

We now give a detailed sketch of the proof of Theorem~\ref{thm: mainM}. We begin by briefly introducing the mixing properties of Langevin dynamics in this multi-modes setting.
In the multi-modes case, on the weighted space~$L^2(\pi_{\epsilon})$ the operator~$L_\epsilon$ is self-adjoint and has a discrete spectrum with eigenvalues
\begin{equation}
  0 = \lambda_{1, \epsilon}
    < \lambda_{2, \epsilon}\leq\dots\leq \lambda_{J,\epsilon}
    \leq \lambda_{J+1, \epsilon}
    \cdots
\end{equation}
with corresponding~$L^2(\pi_\epsilon)$ normalized eigenfunctions~$\psi_{1, \epsilon}$, $\psi_{2, \epsilon}$, etc.
In the situation where~$U$ has~$J$ wells, for every~$\gamma < \hat \gamma$ there exists constants~$C_\gamma$ and~$\Lambda$
(independent of~$\epsilon$)
such that
\begin{equation}
  \lambda_{j,\epsilon}\leq C_\gamma \exp\paren[\Big]{-\frac{\gamma}{\epsilon}}\quad \forall 2\leq j\leq J
  \quad\text{and}\quad
  \lambda_{i,\epsilon}\geq \Lambda,\quad \forall i\geq J+1
  .
\end{equation}
The above spectral decomposition implies that in the low temperature regime, instead of estimating the error of the second eigenfunctions as in the double-well setting, here we need to control the Monte Carlo error of all the eigenfunctions corresponding to the low-lying eigenvalues. Precisely, we have the following multi-modes version of Lemma~\ref{l:langevinError}, whose proof is analogous to that of Lemma~\ref{l:langevinError}.

\begin{lemma}\label{l:langevinErrorM}
  Assume that for each~$i \in \set{1, \dots, N}$,~$Y_{\epsilon,0}^{i}\in K$, where the subset $K$ defined as in~\eqref{eq:defK} with $\hat{\gamma}$ defined by~\eqref{e:gammaHatDef}. Then  for any bounded test function~$h$, with probability larger than or equal to  \begin{equation}
    1-2\exp\paren[\Big]{-\frac{2Na^2}{\norm{h}^2_{\osc}}}
\end{equation}
  we have
  \begin{align}
    \abs[\Big]{ \frac{1}{N}\sum_{i=1}^{N}h(Y^{i}_{\epsilon, T})-\int h\pi_{\epsilon}\,d x}
    &\leq  \sum_{j=2}^{J}e^{-\lambda_{j, \epsilon} T} \abs[\Big]{\int h\psi_{j,\epsilon}\pi_{\epsilon} \, dx} \cdot\abs[\Big]{\frac{1}{N}\sum_{i=1}^{N}\psi_{j,\epsilon}\one_{K}(Y^{i}_{\epsilon,0})}
    \\
    &
    + \mathcal E_{\epsilon,T}(h)  +a
\end{align}
where $\mathcal E_{\epsilon,T}(h)$ is defined as in~\eqref{eq: def of tilde epsilon}.
\end{lemma}  

Next, we introduce the multi-modes versions of the remaining preparing lemmas. Among those, Lemma~\ref{cor:poverpi},~\ref{lem:boundPsi},~\ref{l:rebalancing} and~\ref{lem:globalmixing} remain unchanged. Their proofs and the proof of the prerequisite lemmas are almost identical to those in the double-well setting. It remains to present the multi-modes versions of Lemma~\ref{lem:baseCasePsi} and Lemma~\ref{l:iteration}. In particular, Lemma~\ref{lem:baseCasePsi} becomes the following. It is straightforward to check that its proof is almost identical to that of Lemma~\ref{lem:baseCasePsi}.
\begin{lemma} \label{lem:baseCasePsiM} Fix~$\delta>0, \theta\in (0,\frac13)$. Choose~$N$ as in \eqref{eq:Nchoice} and let~$\eta_{\mathrm{cr}}$ be defined as in \eqref{e:CriEtaLT}. Let~$T$ satisfy \eqref{eq:NTglobal}. Then for all $j=2,\dots,J$,
\begin{equation}
    \P\paren[\Big]{\abs[\Big]{\frac{1}{N}\sum_{i=1}^{N}\psi_{j,k_{\mathrm{cr}}}\one_{K}(X^{i}_{k_{\mathrm{cr}},0})}< (C_r+1)\delta}\geq 1-3\theta.
\end{equation}
   
\end{lemma}

The multi-modes version of Lemma~\ref{l:iteration} is stated as the following. The proof is almost identical to that of Lemma~\ref{l:iteration}, aside from a much more tedious computation. 
\begin{lemma}\label{l:iterationM}
  Choose~$M$ as in~\eqref{e:MTN} and~$\eta_k$ as in~\eqref{e:chooseEtaK}. Fix~$\theta\in (0,1)$.
  Define constant~$C_K\equiv C_K(\alpha,\hat{\gamma})$ as~\eqref{eq:defCK}. There exist dimensional constants~$C_{\alpha} = C_{\alpha}(\alpha,U)>0$ and~$C_{\mathrm{tem}}= C_{\mathrm{tem}}(\alpha,U)>0$ such that for any~$\delta>0$, if~$N,T$ satisfies \begin{align}
 N&\geq 128C^2_r(C_{\psi}+1)^2 \log(J)\frac{M^2}{\delta^2}\log\paren[\Big]{\frac{8M}{\theta}},\\
 T&\geq \max\Bigl\{C_{\alpha}\Big(\log\paren[\Big]{\frac{M}{\theta}}+\log N+\log\paren[\Big]{\frac{1}{\delta}}+\frac{1}{\eta}+\log(J)\Big), \frac{c}{d}\log (4d)\Bigr\},
  \end{align}
  and let~$\eta_{\mathrm{cr}}$ be defined as in~\eqref{e:CriEtaLT}
, then for each~$ k_{\mathrm{cr}}\leq k\leq M-1$, 
  with probability~\eqref{eq:thetakLB}, we have~\eqref{eq:Xk+1inK}
    and
     \begin{equation}\label{e:iterationM}
     \max_{j=2,\dots,J}\abs[\Big]{\frac{1}{N}\sum_{i=1}^{N}\psi_{j,k+1}\one_{K}(X_{k+1,0}^i)}\leq \beta_k \max_{j=2,\dots,J}\abs[\Big]{\frac{1}{N}\sum_{i=1}^{N}\psi_{j,k}\one_{K}(X_{k,0}^i)}+c_k. 
  \end{equation}
  Here the constants~$\beta_k$,~$c_k$ are such that for every~$k$ such that~$\eta_k$ satisfies \eqref{e:CriEtaLT} we have~\eqref{e:CBeta} holds
  for some dimensional constant~$C_{\beta}>1$ (independent of~$\alpha$,~$\delta$).
\end{lemma}

The proof of Theorem~\ref{thm: mainM} then follows similarly as the above-shown proof of Theorem~\ref{thm: main}, where the use of Lemma~\ref{l:langevinError},~\ref{l:iteration} and ~\ref{lem:baseCasePsi} are replaced by Lemma~\ref{l:langevinErrorM},~\ref{l:iterationM} and~\ref{lem:baseCasePsiM}, respectively. After a straightforward but a tedious calculation, one can show that an upper bound of~$C_{\beta}$ actually depends exponentially on~$J$.   As a result, the constant~$C_N$ depends exponentially on~$J$ and~$C_T$ is linear in~$J$. The proof of all the preparing lemmas are analogous to their double-modes version and are thus omitted.

\section{Error estimates for the resampling and the Langevin Dynamics (Lemma~\ref{l:langevinError},~\ref{cor:poverpi} and~\ref{l:rebalancing})} \label{sec:langevin} 
In this section we prove Lemma~\ref{l:langevinError},~\ref{cor:poverpi} and~\ref{l:rebalancing}. 
The proof of Lemma~\ref{l:langevinError} is based on a spectral decomposition, and is presented in Section~\ref{s:langevinError}, below.
The proof of Lemma~\ref{cor:poverpi} is based on a pointwise estimate of the transition density and is presented in Section~\ref{s:pBound}, below. At last, we prove Lemma~\ref{l:rebalancing} in Section~\ref{s:resample}.

\subsection{The Monte Carlo error in the Langevin Step (Lemma~\ref{l:langevinError})} \label{s:langevinError}

Lemma~\ref{l:langevinError} resembles \cite[Lemma 4.6]{han2025polynomialcomplexitysamplingmultimodal} in that both quantify the Monte Carlo error, and the proof here also relies on spectral decomposition in a similar spirit. But~\cite[Lemma 4.6]{han2025polynomialcomplexitysamplingmultimodal} does not apply to our setting. The main reason is that the estimate in \cite[Lemma 4.6]{han2025polynomialcomplexitysamplingmultimodal} requires a uniform bound on the marginal error, independent of the particles’ initial positions. In our unbounded domain, such a bound becomes infinite. Therefore, we instead derive a mixing estimate that depends explicitly on the Langevin starting point.

To this end, we first establish the following preliminary result on the transition density, showing that the transition density~$p_{\epsilon,t}(x,\cdot)$ of~$Y_{\epsilon,t}$, starting from~$x$, becomes very close to a measure whose shape resembles that of the second eigenfunction. This closeness can be quantified in terms of the time~$t$ and the starting point~$x$, implying that local well-mixing can be achieved after a relatively short time if start in a good region. Precisely, we have the following lemma.
\begin{lemma}\label{lem:TDquick}
    Let~$p_{\epsilon,t}(x,\cdot)$ be the transition density of~$Y_{\epsilon,t}$ starting at~$x$, which is defined as in \eqref{eq:defp}, then
    \begin{equation}\label{eq:TDquick}
    \norm[\Big]{\frac{p_{\epsilon,t}(x,\cdot)}{\pi_{\epsilon}(\cdot)}-1 -e^{-\lambda_{2,\epsilon}t}\psi_{2,\epsilon}(x)\psi_{2,\epsilon}(\cdot)}_{L^2(\pi)}\leq e^{-\Lambda t/2}\sqrt{\frac{p_{\epsilon,t}(x,x)}{\pi_{\epsilon}(x)}-1}\,.
\end{equation}
\end{lemma}
\begin{proof}[Proof of Lemma~\ref{lem:TDquick}]
Since in the proof~$\epsilon$ is fixed, for simplicity of presentation, we slightly abuse notation in this proof and omit the subscript~$\epsilon$.

From Chapman-Kolmogorov equation, we have that for~$x,y\in\R^d$,~$t>s>0$,
\begin{equation}\label{eq:CK}
    p_t(x,y)=\int p_s(x,z)p_{t-s}(z,y)\, d z.
\end{equation}
Recall that the stationary distribution~$\pi$ satisfies the detailed balance equations due to the reversibility. That is, for~$y,z\in\mathbb{R}^d$,~$t>s> 0$, we have that
\begin{equation}\label{eq:DB}
    \pi(z)p_{t-s}(z,y)=\pi(y)p_{t-s}(y,z).
\end{equation}
Since~$\pi>0$, we obtain that
\begin{equation}
    \frac{p_t(x,y)}{\pi(y)}\overset{\eqref{eq:CK}}{=}\int  \frac{p_s(x,z)}{\pi(y)} p_{t-s}(z,y)\, d z \overset{\eqref{eq:DB}}{=}\int \frac{p_s(x,z)}{\pi(z)}p_{t-s}(y,z)\, d z,
\end{equation}
which immediately gives that
\begin{equation}\label{eq:trans1}
     \frac{p_t(x,y)}{\pi(y)}-1 =\int \paren[\Big]{\frac{p_s(x,z)}{\pi(z)}-1}p_{t-s}(y,z)\, d z.
\end{equation}
On the other hand, we know that for any~$t>0$, the equation \eqref{e:fOverPiNu} implies that we have the following spectral decomposition of~$\frac{p_t(y,z)}{\pi(z)}$,
\begin{equation}\label{eq:decomP}
   \frac{p_t(y,z)}{\pi(z)}=1+\sum_{k\geq 2}e^{-\lambda_k t}\psi_{k}(y)\psi_k(z).
\end{equation}
Therefore,
\begin{align}
    \frac{p_t(x,y)}{\pi(y)}-1 \:&\overset{\mathclap{\eqref{eq:trans1}}}{=}\int \paren[\Big]{\frac{p_s(x,z)}{\pi(z)}-1}\frac{p_{t-s}(y,z)}{\pi(z)}\pi(z)\, d z\\
    &\overset{\mathclap{\eqref{eq:decomP}}}{=}\int \paren[\Big]{\frac{p_s(x,z)}{\pi(z)}-1}\pi(z)\, d z\\
    &+e^{-\lambda_2 (t-s)}\psi_2(y)\int \paren[\Big]{\frac{p_s(x,z)}{\pi(z)}-1}\psi_2(z)\pi(z)\, d z\\
    &+\sum_{k\geq 3}e^{-\lambda_k (t-s)}\psi_{k}(y) \int\paren[\Big]{\frac{p_s(x,z)}{\pi(z)}-1}\psi_k(z)\pi(z)\, d z.  \label{eq:PoverPi}
\end{align}
Notice that 
\begin{equation}\label{eq:term2}
   \int \paren[\Big]{\frac{p_s(x,z)}{\pi(z)}-1}\psi_2(z)\pi(z)\, d z =\int p_s(x,z)\psi_2(z)\, d z = e^{-\lambda_2 s}\psi_2(x). 
\end{equation}
Plugging \eqref{eq:term2} into \eqref{eq:PoverPi} yields,
\begin{equation}\label{eq:decom2}
    \norm[\Big]{\frac{p_t(x,\cdot)}{\pi(\cdot)}-1 -e^{-\lambda_2t}\psi_2(x)\psi_2(\cdot)}_{L^2(\pi)}\leq e^{-\Lambda (t-s)}\norm[\Big]{\frac{p_s(x,\cdot)}{\pi(\cdot)}-1}_{L^2(\pi)}.
\end{equation}
where~$\Lambda$ is defined in \eqref{eq: egvalLEpsilon}. 

We proceed to estimate~$\norm{p_s(x,\cdot)/\pi(\cdot)-1}_{L^2(\pi)}$. Notice that
\begin{align}
    \norm[\Big]{\frac{p_s(x,\cdot)}{\pi(\cdot)}-1}_{L^2(\pi)}\,&=\int \paren[\Big]{\frac{p_s(x,z)}{\pi(z)}-1}^2\pi(z)\, d z\\
    &=\int \paren[\Big]{\frac{p_s(x,z)}{\pi(z)}}^2\pi(z)\, d z-1\\
    &\overset{\mathclap{\eqref{eq:DB}}}{=}\int \frac{p_s(x,z)}{\pi(z)}\frac{p_s(z,x)}{\pi(x)}\pi(z)\, d z-1\\
    &\overset{\mathclap{\eqref{eq:CK}}}{=}\quad\frac{p_{2s}(x,x)}{\pi(x)}-1. \label{eq:psL2}
\end{align}
Choosing~$s=t/2$ and plugging \eqref{eq:psL2} into \eqref{eq:decom2} finishes the proof.
\end{proof}

Notice that on the right hand side of \eqref{eq:TDquick}, the convergence rate in time~$t$ does not decrease when~$\epsilon$ becomes small. Lemma~\ref{lem:TDquick} describes the quick convergence of~$p_{\epsilon,t}$ towards the local modes (metastable state). We can then prove Lemma~\ref{l:langevinError}.
\begin{proof}[Proof of Lemma~\ref{l:langevinError}]
Let~$\E_0$ denote the conditional expectation given the~$\sigma$-algebra generated by~$\{Y_{\epsilon,0}^i, i=1,\dots,N\}$. Observe that
\begin{equation}\label{e:errhstep1}
    \frac{1}{N}\sum_{i=1}^{N}h(Y^{i}_{\epsilon, T})-\int h\pi_{\epsilon}\,d x=I_1+I_2,
\end{equation}
    where
    \begin{align}
            I_1&\defeq \frac{1}{N}\sum_{i=1}^{N}h(Y^{i}_{\epsilon,T})-\frac{1}{N}\sum_{i=1}^{N}\E_0h(Y^{i}_{\epsilon,T}),\\
            I_2&\defeq \frac{1}{N}\sum_{i=1}^{N}\E_0h(Y^{i}_{\epsilon,T})-\int h\pi_{\epsilon}\,d x.
    \end{align}
\restartsteps
  \step[Bound~$I_1$] Notice that after conditioning on~$\{Y_{\epsilon,0}^1,...,Y_{\epsilon,0}^{N}\}$, the random variables~$Y^{i}_{\epsilon,T}$ are independent. Hence by Hoeffding's inequality,
  \begin{equation}\label{eq:I1}
      \P_0(|I_1|\geq a)\leq 2\exp\paren[\Big]{-\frac{2Na^2}{\norm{h}^2_{\osc}}}.
  \end{equation}

\step[Estimate~$I_2$]
For each fixed~$x$, it can be directly checked that
\begin{align}
    \int p_{\epsilon,T}(x,y)h(y)\, d y&=\int \frac{p_{\epsilon,T}(x,y)}{\pi_{\epsilon}(y)}h(y)\pi_{\epsilon}(y)\, d y\\
    &=\int h\pi_{\epsilon}\, d y+ e^{-\lambda_{2,\epsilon}t}\paren[\Big]{\int h\psi_{2,\epsilon}\pi_{\epsilon}\,dy}\psi_{2,\epsilon}(x)\\
    &+\int \paren[\Big]{\frac{p_{\epsilon,T}(x,y)}{\pi_{\epsilon}(y)}-1-e^{-\lambda_{2,\epsilon}T}\psi_{2,\epsilon}(x)\psi_{2,\epsilon}(y)}h(y)\pi_{\epsilon}(y)\, d y\label{eq:f0perpProj}.
\end{align}
where the third term
\begin{align}
   \MoveEqLeft \abs[\bigg]{\int \paren[\Big]{\frac{p_{\epsilon,T}(x,y)}{\pi_{\epsilon}(y)}-1-e^{-\lambda_{2,\epsilon}T}\psi_{2,\epsilon}(x)\psi_{2,\epsilon}(y)}h(y)\pi_{\epsilon}(y)\, d y}\\
    &\overset{\mathclap{\eqref{eq:TDquick}}}{\leq} \norm{h}_{L^2(\pi_{\epsilon})}e^{-\Lambda T/2}\sqrt{\frac{p_{\epsilon,T}(x,x)}{\pi_{\epsilon}(x)}-1}. \label{eq:I2s}
\end{align}

Then
\begin{align}
|I_2|\:&\overset{\mathclap{\eqref{eq:f0perpProj},\eqref{eq:I2s}}}{\leq} \:  e^{-\lambda_{2, \epsilon} T} \abs[\Big]{\int_{\mathbb T^d}h\psi_{2,\epsilon}\pi_{\epsilon} \, dx} \cdot\abs[\Big]{\frac{1}{N}\sum_{i=1}^{N}\psi_{2,\epsilon}(Y^{i}_{\epsilon,0})}\\
    &+\frac{1}{N} \norm{h}_{L^2(\pi_{\epsilon})} e^{-\Lambda t/2}\sum_{i=1}^{N} \sqrt{\frac{p_{\epsilon,T}(Y_{\epsilon,0}^i,Y_{\epsilon,0}^i)}{\pi(Y_{\epsilon,0}^i)}-1}\\
    &\leq e^{-\lambda_{2, \epsilon} T} \abs[\Big]{\int h\psi_{2,\epsilon}\pi_{\epsilon} \, dx} \cdot\abs[\Big]{\frac{1}{N}\sum_{i=1}^{N}\psi_{2,\epsilon}(Y^{i}_{\epsilon,0})}\\
    &+\norm{h}_{L^2(\pi_{\epsilon})} e^{-\Lambda t/2}\max_{x\in K_{\epsilon}}\sqrt{\frac{p_{\epsilon,T}(x,x)}{\pi(x)}-1}. \label{eq:I2}
\end{align}

\step
We conclude from the above steps that with probability larger than or equal to \eqref{e:LProb}, we have
\begin{multline}
  \abs[\Big]{ \frac{1}{N}\sum_{i=1}^{N}h(Y^{i}_{\epsilon, T})-\int h\pi_{\epsilon}\,d x} 
        \overset{\eqref{e:errhstep1}}{\leq}  \abs{I_1}+\abs{I_2}\\
\overset{\eqref{eq:I1},\eqref{eq:I2}}{\leq}e^{-\lambda_{2, \epsilon} T} \abs[\Big]{\int_{\mathbb T^d}h\psi_{2,\epsilon}\pi_{\epsilon} \, dx} \cdot\abs[\Big]{\frac{1}{N}\sum_{i=1}^{N}\psi_{2,\epsilon}(Y^{i}_{\epsilon,0})}+a+ \mathcal E_{\epsilon,T}(h)\\
=e^{-\lambda_{2, \epsilon} T} \abs[\Big]{\int_{\mathbb T^d}h\psi_{2,\epsilon}\pi_{\epsilon} \, dx} \cdot\abs[\Big]{\frac{1}{N}\sum_{i=1}^{N}\psi_{2,\epsilon}\one_{K}(Y^{i}_{\epsilon,0})}+a+ \mathcal E_{\epsilon,T}(h)
\end{multline}
where the last equality follows from the assumption that~$Y_{\epsilon,0}^{k}\in K$ for each~$i\in\set{1,\dots,N}$.
\end{proof}

\subsection{Estimate of transition kernel (Lemma~\ref{cor:poverpi})}\label{s:pBound}

In this section, we prove Lemma~\ref{cor:poverpi}.  
Before the proof, we need  apriori estimate on pointwise bound of~$p_{\epsilon,t}$, which is stated in the following lemma. The proof employs Nash’s argument~\cite{nashPara}, based on the Chapman–Kolmogorov equation and Fourier transform and combined with an absorbing argument. The original result in~\cite{nashPara}  considers general operators in the divergence form and the dependence on the coefficients is not explicit, so we need to do extra work to obtain the estimate for~$p_{\epsilon}$ with a precise~$\epsilon$ dependence. 
\begin{lemma}\label{lem:EstGreenUB}
Assume that~$U$ satisfies Assumption \ref{a:criticalpts}. Let~$p_t(x,y)$ be the transition kernel, then for every~$x\in\mathbb{R}^d$ and~$t>0$,
\begin{equation}\label{eq:ptxxEst}
     p_{\epsilon,t}(x,x) \leq
    \paren[\bigg]{\frac{2\epsilon C_d}{c}\paren[\Big]{1-\exp\paren[\Big]{-\frac{ct}{d}}}}^{-\frac{d}{2}}.
\end{equation}
Here $C_d$ is a dimensional constant.
\end{lemma}
\begin{proof}[Proof of Lemma~\ref{lem:EstGreenUB}]
By Chapman-Kolmogorov equation
\begin{equation}
   p_t(x,x) = \int p_{\frac{t}{2}}(x,z)  p_{\frac{t}{2}}(z,x)\,  d z
\end{equation}
and thus by Cauchy-Schwarz inequality,
\begin{align}
    (p_t(x,x))^2 &\leq \int \paren[\Big]{p_{\frac{t}{2}}(x,z)}^2\, d z \cdot \int \paren[\Big]{p_{\frac{t}{2}}(z,x)}^2\, d z. 
\end{align}
The proof of Lemma consists of two steps, estimating the~$L^2$ norm of the forward density and the backward density, respectively.
\restartsteps
\step[Forward density]
    By Kolmogorov forward equation, for fixed~$x$,  the forward density~$p(\cdot)=p_t(x,\cdot)$ satisfies the following equation.
    \begin{equation}\label{eq:FP}
        \partial_t p=-\mathrm{div}\paren{(\nabla U) p}+\epsilon \Delta p.
    \end{equation}

For a fixed~$x$, let~$T=T(y,t)=p_t(x,y)$ and 
\begin{equation}\label{eq:defE}
    E\defeq \int \paren[\big]{T(y,t)}^2\, d y.
\end{equation}
By a direct computation,
\begin{align}
    \partial_t E \quad&\overset{\mathclap{\eqref{eq:defE}}}{=}\quad 2\int T\paren{\partial_t T}\, d y
\overset{\eqref{eq:FP}}{=} 2\int T\paren{\mathrm{div}\paren{(\nabla U) T}+\epsilon \Delta T}\, d y\\
&=2\int \paren{\nabla T\cdot\nabla U} T\,dy-2\epsilon\int \abs{\nabla T}^2\, d y\\
&=\int \nabla (T^2)\cdot \nabla U\, d y -2\epsilon\int \abs{\nabla T}^2\, d y\\
&=-\int (\Delta U)T^2\, d y -2\epsilon\int \abs{\nabla T}^2\, d y. \label{eq:Et}
\end{align}
For simplicity of notation, we further let
\begin{equation}
    u(y)\defeq T(y,t).
\end{equation}
The Fourier transform of~$u(y)$ is 
\begin{equation}
    v(z)=(2\pi)^{-\frac{d}{2}}\int e^{i y\cdot z}u(y)\, dy.
\end{equation}
This has the familiar property
\begin{equation}\label{eq:F1}
    \int |v|^2 \, d z =\int |u|^2\, d y.
\end{equation}
The Fourier transform of~$\partial u/\partial y_k$ is~$iz_ku$, hence
\begin{equation}
\int \abs{\partial u/\partial y_k}^2\, dy=\int |z_k|^2 |v|^2\, d z
\end{equation}
and
\begin{equation}\label{eq:F2}
    \int |\nabla u|^2\, d y=\int |z|^2|v|^2\, d z.
\end{equation}
Finally,
\begin{equation}\label{eq:F3}
    |v(z)|\leq (2\pi)^{-\frac{d}{2}}\int |e^{iy\cdot z}|\cdot |u|\, d y =(2\pi)^{-\frac{d}{2}}\int |u|\, d y.
\end{equation}
For any~$\rho>0$, using the formula for the volume of a~$d$-sphere,
\begin{equation}
    \int_{|z|\leq \rho}|v|^2\, d z\overset{\eqref{eq:F3}}{\leq}\paren[\Big]{\frac{\pi^{\frac{d}{2}}\rho^d}{(d/2)!}}\paren[\big]{(2\pi)^{-\frac{d}{2}}\int |u|\, d y}^2= a \rho^d
\end{equation}
where
\begin{equation}\label{eq:defa}
    a\defeq \frac{\pi^{\frac{d}{2}}}{(d/2)!}\paren[\big]{(2\pi)^{-\frac{d}{2}}\int |u|\, d y}^2.
\end{equation}
On the other hand,
\begin{equation}
    \int_{|z|>\rho}|v|^2\, d z\leq \int_{|z|< \rho}\abs[\Big]{\frac{z}{\rho}}^2|v|^2\, d z \overset{\eqref{eq:F2}}{\leq}\frac{1}{\rho^2}\int |\nabla u|^2\, d y= b \rho^{-2}
\end{equation}
where
\begin{equation}\label{eq:defb}
    b\defeq \int |\nabla u|^2\, d y.
\end{equation}
Choosing 
\begin{equation}\label{eq:ChooseRho}
    \rho = \paren[\Big]{\frac{2b}{ad}}^{\frac{1}{d+2}},
\end{equation}
we obtain that
\begin{align}
   \MoveEqLeft \int |u|^2\, d y \overset{\mathclap{\eqref{eq:F1}}}{=} \int |v|^2\, d z = \int_{|z|\leq \rho}|v|^2\, d z+ \int_{|z|> \rho}|v|^2\, d z\\
    &\overset{\mathclap{\eqref{eq:ChooseRho},\eqref{eq:defa},\eqref{eq:defb}}}{\leq}\qquad \quad\paren[\Big]{\frac{\pi^{\frac{d}{2}}}{(d/2)!}}^{\frac{2}{d+2}}\paren[\Big]{(2\pi)^{-\frac{d}{2}}\int |u|\, d y}^{\frac{4}{d+2}} \paren[\Big]{\int |\nabla u|^2\, d y}^{\frac{d}{d+2}}\paren{1+\frac{d}{2}}\paren[\Big]{\frac{2}{d}}^{\frac{d}{d+2}}.
\end{align}
Hence, using the identity that~$\int |u|\, dy\equiv 1$, we obtain that
\begin{align}
    \int |\nabla u|^2\, d y&\geq \paren[\Big]{\frac{4\pi d}{d+2}} \paren[\Big]{\frac{(d/2)!}{1+\frac{d}{2}}}^{\frac{2}{d}}\paren[\Big]{\int |u|\, dy}^{-\frac{4}{d}}\paren[\Big]{\int |u|^2\, dy}^{1+\frac{2}{d}}\\
    &=C_d\paren[\Big]{\int |u|^2\, dy}^{1+\frac{2}{d}} \label{eq:nablausq}
\end{align}
where
\begin{equation}
    C_d\defeq \paren[\Big]{\frac{4\pi d}{d+2}} \paren[\Big]{\frac{(d/2)!}{1+\frac{d}{2}}}^{\frac{2}{d}}.
\end{equation}

According to the assumption \eqref{eq:Blap} that~$\Delta U\geq -c$, we have that
\begin{equation}
    \partial_t E \overset{\eqref{eq:Et},\eqref{eq:nablausq}}{\leq} c E -2\epsilon C_dE^{1+\frac{2}{d}}.
\end{equation}
Multiply both sides by~$E^{-1-\frac{2}{d}}$, we obtain
\begin{equation}
    E^{-1-\frac{2}{d}}\partial_t E \leq c E^{-\frac{2}{d}} -2\epsilon C_d.
\end{equation}
Let~$F=E^{-\frac{2}{d}}$, a straightforward calculation yields that
\begin{equation}
    \partial_t F +\frac{2c}{d}F\geq \frac{4\epsilon C_d}{d}.
\end{equation}
Solving the above gives that
\begin{equation}
    F\geq \frac{2\epsilon C_d}{c}\paren[\Big]{1-\exp(-\frac{2c}{d}t)}
\end{equation}
which implies that
\begin{equation}\label{eq:E1}
    E\leq \paren[\bigg]{\frac{2\epsilon C_d}{c}\paren[\Big]{1-\exp(-\frac{2c}{d}t)}}^{-\frac{d}{2}}.
\end{equation}

\step[Backward density]
By Kolmogorov backward equation, for fixed~$y$, the density~$p(\cdot)=p(\cdot,0;y,t)$,~$t>0$ satisfies the following equation
\begin{equation}
    \partial_t p =-\nabla U\cdot \nabla p+\epsilon\Delta p.
\end{equation}

For a fixed~$y$, let~$T=T(x,t)=p(x,0; y, t)$ and 
\begin{equation}\label{eq:defE'}
    E\defeq \int \paren[\big]{T(x,t)}^2\, d x.
\end{equation}
By a direct computation,
\begin{align}
    \partial_t E \quad&\overset{\mathclap{\eqref{eq:defE}}}{=}\quad 2\int T\paren{\partial_t T}\, d y
\overset{\eqref{eq:FP}}{=} 2\int T\paren{-\nabla U \cdot\nabla T+\epsilon \Delta T}\, d y\\
&=-\int \nabla (T^2)\cdot \nabla U\, d y -2\epsilon\int \abs{\nabla T}^2\, d y\\
&=\int (\Delta U)T^2\, d y -2\epsilon\int \abs{\nabla T}^2\, d y. \label{eq:Et'}
\end{align}
For simplicity of notation, we further let
\begin{equation}
    u(x)\defeq T(x,t).
\end{equation}
Following a similar argument to the above step, according to the assumption \eqref{eq:Blap} that~$\Delta U\leq c$, we have that
\begin{equation}
    \partial_t E \overset{\eqref{eq:Et},\eqref{eq:nablausq}}{\leq} c E -2\epsilon C_dE^{1+\frac{2}{d}}.
\end{equation}
which by a similar argument as above gives that
\begin{equation}\label{eq:E2}
    E\leq \paren[\bigg]{\frac{2\epsilon C_d}{c}\paren[\Big]{1-\exp(-\frac{2c}{d}t)}}^{-\frac{d}{2}}.
\end{equation}

\step 
By Chapman-Kolmogorov equation
\begin{equation}
   p_t(x,x) = \int p_{\frac{t}{2}}(x,z)  p_{\frac{t}{2}}(z,x)\,  d z
\end{equation}
and thus by Cauchy-Schwarz inequality,
\begin{align}
    (p_t(x,x))^2 &\leq \int \paren[\Big]{p_{\frac{t}{2}}(x,z)}^2\, d z \cdot \int \paren[\Big]{p_{\frac{t}{2}}(z,x)}^2\, d z \\
    &\overset{\mathclap{\eqref{eq:E1},\eqref{eq:E2}}}{\leq}~\quad\paren[\bigg]{\frac{2\epsilon C_d}{c}\paren[\Big]{1-\exp(-\frac{c}{d}t)}}^{-d}.
\end{align}
Taking square root on both sides finishes the proof of \eqref{eq:ptxxEst}.
\end{proof}

With Lemma~\ref{lem:EstGreenUB} in hand, the proof of Lemma~\ref{cor:poverpi} becomes straightforward, as shown below.
\begin{proof}[Proof of Lemma~\ref{cor:poverpi}]
We know that
\begin{equation}
    \pi_{\epsilon}(x)=\frac{1}{Z_{\epsilon}}\exp\paren[\Big]{-\frac{U(x)}{\epsilon}},
\end{equation}
which implies that
\begin{equation}\label{eq:ppi}
        \frac{p_{\epsilon,t}(x,x)}{\pi_{\epsilon}(x)}\leq Z_{\epsilon}\exp\paren[\Big]{\frac{U(x)}{\epsilon}}\paren[\bigg]{\frac{2\epsilon C_d}{c}\paren[\Big]{1-\exp(-\frac{c}{d}t)}}^{-\frac{d}{2}}.
    \end{equation}
    Now we upper bound~$Z_{\epsilon}$. Notice that by Assumption~\ref{a:massRatioBound}
    \begin{equation}
        Z_{\epsilon}=\paren{\pi_{\epsilon}(\Omega_1)+\pi_{\epsilon}(\Omega_2)}Z_{\epsilon}\leq (1+C_m^2)\pi_{\epsilon}(\Omega_1)Z_{\epsilon}.
    \end{equation}
By \cite[Equation (2.16)]{MenzSchlichting14},
\begin{equation}\label{eq:Zepsilon}
    \pi_{\epsilon}(\Omega_i)Z_{\epsilon}\leq C \frac{(2\pi \epsilon)^{\frac{d}{2}}}{\sqrt{\det\nabla^2 U(x_{\min,i})}}\exp\paren[\Big]{-\frac{U(x_{\min,i})}{\epsilon}}.
\end{equation}
Here~$C$ is a dimensional constant,~$C\to 1$ as~$\epsilon\to 0$. When~$\epsilon <1$, we can take~$C$ as a dimensional constant. Taking~$i=1$ in \eqref{eq:Zepsilon} gives that
\begin{equation}\label{eq:Zepsilon1}
    \pi_{\epsilon}(\Omega_1)Z_{\epsilon}\leq C \frac{(2\pi \epsilon)^{\frac{d}{2}}}{\sqrt{\det\nabla^2 U(x_{\min,1})}}
\end{equation}
where we use the Assumption~\ref{a:criticalpts}.
Plugging \eqref{eq:Zepsilon1} into \eqref{eq:ppi} yields that for all~$\epsilon\in (0,1)$, there exists a dimensional constant~$C_p\equiv C_p(d,U)$ such that
\eqref{eq:poverpi} holds. Then \eqref{eq:poverpiK} follows by noticing that for $x\in K$,
\begin{equation}
U(x)\overset{\eqref{eq:defB},\eqref{eq:defK}}{\leq} U(x_{\min,2})+\hat{\gamma}= U(x_{\min,2})+U(s_{1,2})-U(x_{\min,2})=\hat{U}.\qedhere
\end{equation}
\end{proof}

\subsection{The Rebalancing Error (Lemma~\ref{l:rebalancing})} \label{s:resample}

In this section, we prove Lemma Lemma~\ref{l:rebalancing}, which is very similar to \cite[Lemma 3.3]{han2025polynomialcomplexitysamplingmultimodal} and the proof technique is standard and widely used in the context such as sequential Monte Carlo and importance sampling, see e.g., \cite[Chapter 11]{Chopin04} and \cite[Chapter 3]{sanz2024book}. 

Notice that the points~$y^{1}$, \dots,~$y_N$ chosen according to~\eqref{e:PyiEqXj} are identically distributed, but need not be independent.
However, given the points~$x^{1}$, \dots,~$x_N$, the points~$y^{1}$, \dots,~$y_N$ are (conditionally) independent.
The main idea behind the proof of Lemma~\ref{l:rebalancing} is to use conditional independence of~$y^{1}$, \dots,~$y_N$ and Hoeffding's inequality.

\begin{proof}[Proof of Lemma~\ref{l:rebalancing}]
    For simplicity of notation, let
  \begin{equation}\label{eq: defTildeh}
     \mathbf{x}\defeq \{x^1,...,x^{N}\},\quad\text{and}\quad \tilde{h}\defeq h-\int_{\mathcal{X}} hq\, dx
\end{equation}
  and let~$\E_{\mathbf x}$ denote the conditional expectation given the~$\sigma$-algebra generated by~$\mathbf x$.
    By the tower property,
    \begin{align}
  \frac{1}{N} \sum_1^N h(y^{i}) - \int_{\mathcal{X}} hq\, dx 
    =J_1+J_2
    \end{align}
where
\begin{equation}\label{e:J1J2}
        J_1\defeq \frac{1}{N}\sum_{\ell=1}^{N}\E_{\mathbf{x}}\tilde{h}(y^{i})\\
	\quad\text{and}\quad
        J_2\defeq \frac{1}{N}\sum_{\ell=1}^N
      \tilde{h}(y^{i})-\frac{1}{N}\sum_{\ell=1}^{N}\E_{\mathbf{x}}\tilde{h}(y^{i})
      .
\end{equation}
Thus
\begin{equation}\label{e:Ehy1}
     \abs[\Big]{\frac{1}{N} \sum_1^N h(y^{i}) - \int_{\mathcal{X}} hq\, dx}\leq |J_1|+|J_2|.
\end{equation}
Notice that the points~$y^1$, \dots,~$y^N$ are not independent; however, when conditioned on~$\mathbf x$, the points~$y^i$ are independent and identically distributed. Thus by Hoeffding's inequality
  \begin{equation}\label{e:J2}
 \P_{\mathbf{x}}(|J_2|>a)\leq 2\exp\paren{-\frac{2Na^2}{\norm{h}^2_{\osc}}}  .
\end{equation}

To bound~$J_1$, we note
  \begin{align}
    \frac{1}{N}\sum_{\ell=1}^{N}\E_{\mathbf{x}}\tilde{h}(y^{i}) &=\E_{\mathbf{x}}\tilde{h}(y_{1})
	  =\frac{ \sum_{i=1}^N
	 \tilde{h}(x^{i}) \tilde r(x^{i})}{\sum_{i=1}^N \tilde r(x^{i}) }
	 \overset{\eqref{e:rdef}}{=}
	 \frac{ \sum_{i=1}^N
	 \tilde{h}(x^{i}) {r}(x^{i})}{\sum_{i=1}^N {r}(x^{i}) }
    \\
      \label{e:J1sqrt}
      &= J_3 + J_4,
\end{align}
where
\begin{equation}
        J_3=\frac{ \sum_{i=1}^N
	 \tilde{h}(x^{i}) r(x^{i})}{\sum_{i=1}^N r(x^{i}) }\paren[\Big]{1-\frac{1}{N}\sum_{i=1}^N r(x^{i}) },
	 \quad\text{and}\quad
   J_4= \frac{1}{N}\sum_{i=1}^N
	  \tilde{h}(x^{i}) r(x^{i}).
\end{equation}
Clearly 
\begin{equation}\label{e:j3}
  \abs{J_3} \leq \norm{\tilde h}_{L^\infty} \paren[\Big]{1-\frac{1}{N}\sum_{i=1}^N r(x^{i}) }
  .
\end{equation}

Thus with probability larger than or equal to \eqref{e:probResam}
\begin{equation}
  \abs[\Big]{\frac{1}{N} \sum_1^N h(y^{i}) - \int_{\mathcal{X}} hq\, dx}
    \overset{\eqref{e:Ehy1}}\leq 
      |J_1|+|J_2|
  \overset{\eqref{e:J1sqrt},~\eqref{e:J2}}{\leq}
     |J_3|+|J_4|+a.
\end{equation}
Using the definition of~$\tilde h$ in~\eqref{eq: defTildeh} we obtain~\eqref{e:resampling} as desired.
\end{proof}

\section{Apriori estimates}
In this section, we prove some preliminary estimates used in the proof of Lemma~\ref{l:iteration}. In particular, we estimate the  probability that particles fall in set~$K$ in Section~\ref{sec:K} and prove Lemma~\ref{lem:boundPsi} in Section~\ref{sec:boundPsi}.  

\subsection{Probability that particles fall in set \texorpdfstring{$K$}{K}} \label{sec:K}
In this section, we show that the probability that particles escape from the set~$K$ is exponentially small in~$\epsilon$. Thus when the temperature is low, we should expect fewer particles escape from~$K$, which implies that the condition of Lemma~\ref{l:langevinError} holds with high probability.

We begin with a straightforward calculation to estimate the lower bound of~$\pi_{\epsilon}(K)$.

\begin{lemma}\label{lem:piKbound} Fix~$\alpha>0$ and let~$K$ be defined as in \eqref{eq:defK}. 
There exists constant~$C_P\equiv C_P(U,\alpha)$ such that
    \begin{equation}\label{eq:piKbound}
        \pi_{\epsilon}(K)\geq 1-C_P\exp\paren[\Big]{-\frac{C_{K}}{\epsilon}},
    \end{equation}
    where~$C_K$ is the constant defined in \eqref{eq:defCK}
\end{lemma}

\begin{proof}[Proof of Lemma~\ref{lem:piKbound}]
Observe that according to \eqref{eq:defB},
\begin{equation}
    C_K=\frac{1}{(1+\alpha)^{\frac14}}\min_{i=1,2}\inf_{y\in\partial B_i}\paren[\Big]{ U(y)-U(x_{\min,i})}.
\end{equation}
Define the subset~$\tilde{K}\subseteq K$ as 
\begin{equation}\label{eq:defTildeK}
    \tilde{K}\defeq \bigcup_{i=1,2}\set[\Big]{x\in B_i \st U(x)\leq \inf_{y\in\partial B_i}U(y)-C_K}.
\end{equation}
 Assumption~\ref{a:criticalpts} guarantees that~$\tilde{K}$ is a nonempty compact set.
For every~$\epsilon<1$,
\begin{align}
    \frac{\pi_{\epsilon}(K^c)}{\pi_{\epsilon}(K)}&=\frac{\int_{K^c}\exp(-U/\epsilon)\, d x}{\int_{K}\exp(-U/\epsilon)\, d x}\leq \frac{\int_{K^c}\exp(-U/\epsilon)\, d x}{\int_{\tilde{K}}\exp(-U/\epsilon)\, d x}\\
    &\leq \frac{\exp\paren[\Big]{(1-\frac{1}{\epsilon})\inf_{y\in K^c}U(y)}\int_{K^c}\exp(-U)\, d x}{\exp\paren[\Big]{(1-\frac{1}{\epsilon})\sup_{y\in \tilde{K}}U(y)}\int_{\tilde{K}}\exp(-U)\, d x}\\
    &\overset{\mathclap{\eqref{eq:defTildeK}}}{\leq}~ \exp\paren[\Big]{\paren[\big]{1-\frac{1}{\epsilon}}C_K}\frac{\pi_1(K^c)}{\pi_1(\tilde{K})}.
    \label{eq:ratioPiK}
\end{align}

    Notice that
    \begin{equation}
        \pi_{\epsilon}(K)=\frac{\pi_{\epsilon}(K)}{\pi_{\epsilon}(K)+\pi_{\epsilon}(K^c)}=\frac{1}{1+\frac{\pi_{\epsilon}(K^c)}{\pi_{\epsilon}(K)}}\geq 1-\frac{\pi_{\epsilon}(K^c)}{\pi_{\epsilon}(K)}.
    \end{equation}
    Plugging in \eqref{eq:ratioPiK} finishes the proof with
    \begin{equation}
        C_P\defeq \exp(C_K)\frac{\pi_1(K^c)}{\pi_1(\tilde{K})}. \qedhere
    \end{equation}
\end{proof}

Then we estimate the probability that ~$Y_{\epsilon,T}$  fall in set~$K$ when start from~$x$, which we denote as~$\P(Y_{\epsilon,T}\in K| Y_{\epsilon,0}=x)$. Intuitively, if the mixing is sufficient, then we should expect this probability to be large. One may worry that when temperature is low, the global mixing is extremely slow. If the running time is relatively short,  the probability~$\P(Y_{\epsilon,T}\in K| Y_{\epsilon,0}=x)$ might differ from~$\pi_{\epsilon}(K)$ a lot. However, when~$\epsilon$ is small, particles travel into~$K$ quickly and are likely to stay inside, as the following lemma shows. 

\begin{lemma}\label{lem:lbxinKsingle}
For every~$\epsilon\in (0,1)$ and~$x\in K$, we have that
    \begin{equation}
        \P(Y_{\epsilon,T}\in K| Y_{\epsilon,0}=x)\geq 1-C_P(C_{\psi}+1)\exp\paren[\Big]{-\frac{C_{K}}{\epsilon}}-e^{-\Lambda T/2}\sqrt{\frac{p_{\epsilon,T}(x,x)}{\pi_{\epsilon}(x)}-1}.
    \end{equation}
\end{lemma}

\begin{proof}[Proof of Lemma~\ref{lem:lbxinKsingle}]
We know that
\begin{equation}
    \P(X_{\epsilon,T}\in K| X_{\epsilon,0}=x)=\E_{x}\one_{K}(X_{\epsilon,T})=\int \one_{K}(y)p_{\epsilon,T}(x,y)\, d y.
\end{equation}
   Plugging~$h=\one_{K}$ into \eqref{eq:f0perpProj} and \eqref{eq:I2s} gives that 
\begin{align*}
    \MoveEqLeft\abs[\Big]{\int \one_{K}(y)p_{\epsilon,T}(x,y)\, d y -\pi_{\epsilon}(K)-e^{-\lambda_{2,\epsilon}T}\int_{K}\psi_{2,\epsilon}\pi_{\epsilon}\, d y\psi_{2,\epsilon}(x)}\\
    &\leq \norm{\one_{K}}_{L^2(\pi_{\epsilon})}e^{-\Lambda T/2}\sqrt{\frac{p_{\epsilon,T}(x,x)}{\pi_{\epsilon}(x)}-1},
\end{align*}
which implies that
\begin{align}
    \MoveEqLeft\P(X_{\epsilon,T}\in K| X_{\epsilon,0}=x)=\int \one_{K}(y)p_{\epsilon,T}(x,y)\, d y \\
    &\geq \pi_{\epsilon}(K)-e^{-\lambda_{2,\epsilon}T}\abs[\Big]{\int_{K}\psi_{2,\epsilon}\pi_{\epsilon}\, d y}\abs{\psi_{2,\epsilon}(x)}\norm{\one_{K}}_{L^2(\pi_{\epsilon})}\\
    &-e^{-\Lambda T/2}\sqrt{\frac{p_{\epsilon,T}(x,x)}{\pi_{\epsilon}(x)}-1}. \label{eq:pKxy}
\end{align}
Now we bound the three terms on the right hand side of \eqref{eq:pKxy}. The first term is lower bounded by \eqref{eq:piKbound}. For the second term, notice that 
\begin{equation}\label{e:intKpsipi}
    \abs[\Big]{\int_{K}\psi_{2,\epsilon}\pi_{\epsilon}\, d y}=\abs[\Big]{\int_{K^c}\psi_{2,\epsilon}\pi_{\epsilon}\, d y}\leq \pi_{\epsilon}(K^c)\overset{\eqref{eq:piKbound}}{\leq}C_P\exp\paren[\Big]{-\frac{C_{K}}{\epsilon}}.
\end{equation}
where the first equality we use the identity~$\int\psi_{2,\epsilon}\pi_{\epsilon}\, d y\equiv 0$. Thus, the second term
\begin{align}
   \MoveEqLeft e^{-\lambda_{2,\epsilon}T}\abs[\Big]{\int_{K}\psi_{2,\epsilon}\pi_{\epsilon}\, d y}\abs{\psi_{2,\epsilon}(x)}\norm{\one_{K}}_{L^2(\pi_{\epsilon})}\leq \abs[\Big]{\int_{K}\psi_{2,\epsilon}\pi_{\epsilon}\, d y}\abs{\psi_{2,\epsilon}(x)} \\
   &\overset{\mathclap{\eqref{eq:defCpsi},\eqref{eq:piKbound}}}{\leq} C_PC_{\psi}\exp\paren[\Big]{-\frac{C_{K}}{\epsilon}}. \label{e:pKT2}
\end{align}
Plugging \eqref{eq:piKbound} and \eqref{e:pKT2} into \eqref{eq:pKxy} completes the proof.
\end{proof}

By the union bound, Lemma~\ref{lem:lbxinKsingle} immediately gives that
\begin{corollary}\label{cor:lbxinK}
We have that
\begin{align}
    \MoveEqLeft\P(Y_{\epsilon,T}^{i}\in K,\forall i \:|\: Y_{0,T}^{i}\in K,\forall i)\\
    &\geq 1-N\paren[\bigg]{C_P(C_{\psi}+1)\exp\paren[\Big]{-\frac{C_{K}}{\epsilon}}+e^{-\Lambda T/2}\max_{x\in K}\sqrt{\frac{p_{\epsilon,T}(x,x)}{\pi_{\epsilon}(x)}-1}}. \label{eq:lbxinK}
\end{align}
\end{corollary}
Corollary~\ref{cor:lbxinK} suggested that the threshold for low temperature should scale like~$1/\log(N)$, which coincides with our choice of~$\eta_{\mathrm{cr}}$ in \eqref{e:CriEtaLT}.

\subsection{Uniform boundedness of eigenfunctions on subset~$K$ (Lemma~\ref{lem:boundPsi})} \label{sec:boundPsi}

In this section we prove Lemma~\ref{lem:boundPsi}. The proof is very similar to that of~\cite[Lemma 8.2]{han2025polynomialcomplexitysamplingmultimodal} and the main tool is also local maximum principle. The difference is that here the constant~$C_{\psi}$ is~$\alpha$-dependent whereas in~\cite[Lemma 8.2]{han2025polynomialcomplexitysamplingmultimodal} it is not. In the proof, the constant~$C= C(U,d,C_m,\alpha)$ may change from line to line. 

We begin by stating the fact that when~$\epsilon$ is small, the second eigenfunction~$\psi_{2,\epsilon}$ is very close to a linear combination of~$\one_{\Omega_1}$ and~$\one_{\Omega_2}$.
To state the fact precisely, we consider  the subspaces~$E_\epsilon, F_\epsilon \subseteq L^2(\pi_\epsilon)$ defined by
\begin{equation}\label{eq:defEepsFeps}
   F_{\epsilon}\defeq\operatorname{span}\{1,\psi_{2,\epsilon}\},\quad E_{\epsilon}\defeq\operatorname{span}\{\one_{\Omega_1},\one_{\Omega_2}\}.
\end{equation}
We measure closeness of~$\psi_{2, \epsilon}$ to a linear combination of~$\one_{\Omega_1}$ and~$\one_{\Omega_2}$, by measuring the ``distance'' between the subspaces~$E_\epsilon$ and~$F_\epsilon$ defined by
\begin{equation*}
   d(E_{\epsilon}, F_{\epsilon})\defeq \norm{P_{E_{\epsilon}}-P_{E_{\epsilon}}P_{F_{\epsilon}}}=\norm{P_{E_{\epsilon}}-P_{F_{\epsilon}}P_{E_{\epsilon}}}.
\end{equation*}
Here~$P_{E_{\epsilon}}$,~$P_{E_{\epsilon}}$ are the~$L^2(\pi_\epsilon)$ orthonormal projectors onto~$E_{\epsilon}$ and~$F_{\epsilon}$ respectively.
The next result gives an estimate on~$d(E_\epsilon, F_\epsilon)$.

\begin{proposition}[Chapter 8,  Proposition 2.2 of \cite{Kolokoltsov00}]\label{prop_kolo_chap8_prop2.2}
  Let~$\hat \gamma$ be the energy barrier defined in~\eqref{e:gammaHatDef}.
  For any~$\gamma<\hat{\gamma}$,
  there exists a constant~$C_{\gamma}>0$ such that for all~$\epsilon\leq 1$, we have
\begin{equation}\label{e:dEF}
    d(E_{\epsilon}, F_{\epsilon})\leq C_{\gamma}\exp
\Big(\frac{-\gamma}{\epsilon}\Big).
\end{equation}
\end{proposition}
 Proposition~\ref{prop_kolo_chap8_prop2.2} shows the fact that when~$\epsilon$ is small, the second eigenfunction~$\psi_{2,\epsilon}$ is very close to a linear combination of~$\one_{\Omega_1}$ and~$\one_{\Omega_2}$. We then estimate the coefficients of the linear combination. The following discussion is the same as that in~\cite[Proof of Lemma 7.5]{han2025polynomialcomplexitysamplingmultimodal}. For completeness, we repeat them here.

We decompose~$\psi_{2,\epsilon}$ into the sum of the projection into~$E_{\epsilon}$ and~$E^{\perp}_{\epsilon}$, where~$E_{\epsilon}$ is defined in \eqref{eq:defEepsFeps}. Explicitly,
\begin{equation}
    \psi_{2,\epsilon}(x)=\upsilon_{\epsilon}(x)+\sum_{j=1,2}a_{j,\epsilon}\one_{\Omega_j},
\end{equation}
where~$\upsilon_{\epsilon}\in E^{\perp}_{\epsilon}$. Since~$\int \psi_{2,\epsilon}(x)\pi_{\epsilon}(x)\, d x=0$, we have $\sum_{j=1,2}a_{j,\epsilon}\one_{\Omega_j}=0$. Moreover, 
since~$\upsilon_{\epsilon}$ is orthogonal to~$\one_{\Omega_1}$ and~$\one_{\Omega_1}$ in~$L^2(\pi_{\epsilon})$, we define
\begin{equation}\label{eq: ak eq1}
    b_{\epsilon}\defeq \sqrt{\sum_{j=1,2}a_{j,\epsilon}^2\pi_{\epsilon}(\Omega_j)}=\sqrt{1-\|\upsilon_{\epsilon}\|^2_{L^2(\pi_{\epsilon})}}.
\end{equation}
Then we have~$b_{\epsilon}\leq 1$ and the solution of~$a_{j,\epsilon}, j=1,2$ satisfying 
\begin{equation}\label{eq: ak sol}
\abs{a_{j,\epsilon}}\leq \frac{1}{\sqrt{\pi_{\epsilon}(\Omega_j)}}\overset{\eqref{e:massRatioBound}}{\leq}C_m.
\end{equation}
Finally we observe that 
\begin{align}
\|\upsilon_{\epsilon}\|_{L^2(\pi_{\epsilon})}&=\|P_{E_{\epsilon}}^{\perp}(\psi_{2,\epsilon})\|_{L^2(\pi_{\epsilon})}=\|(I-P_{E_{\epsilon}})(\psi_{2,\epsilon})\|_{L^2(\pi_{\epsilon})}\\
       &=\|(P_{F_{\epsilon}}-P_{E_{\epsilon}}P_{F_{\epsilon}})(\psi_{2,\epsilon})\|_{L^2(\pi_{\epsilon})}\leq \|P_{F_{\epsilon}}-P_{E_{\epsilon}}P_{F_{\epsilon}}\|
       = d(E_{\epsilon},F_{\epsilon}). \label{eqVkLessThanDEkFk} 
\end{align}

Now we show that when~$\epsilon$ is small enough,~$\psi_{2,\epsilon}$ is actually close to~$\sum_{i=1,2}a_{i,\epsilon}\one_{\Omega_i}$ in the pointwise sense within the compact set~$K$. Before the proof, we first introduce some auxiliary sets. Define the sets~$\widetilde{B}_i$,~$i\in\set{1,2}$ as
\begin{equation}\label{eq: deftildeB}
  \widetilde{B}_i\defeq \set[\Big]{x\in\Omega_i  \st U(x)-U(x_{\min,i})\leq \frac{\hat \gamma}{(1+\alpha)^{\frac18}}}. 
\end{equation}

According the definition of~$B_i$ and~$\tilde{B}_i$, it is straightforward to see that~$B_i$ is a strict subset of~$\tilde{B}_i$. For~$i=1,2$, define
\begin{equation}
    R_i\defeq \min\{d(\partial \tilde{B}_i, \partial B_i), d(x_{\min,i},\partial B_i)\}.
\end{equation}
where the distance between two subsets~$V_1, V_2$ are defined as in the common sense
\begin{equation}
    d(V_1, V_2)\defeq \inf\set{\abs{x-y}, x\in V_1, y\in V_2}.
\end{equation}
The sets~$\tilde{B}_i$,~$ \partial B_i$ and~$\set{x_{\min,i}}$ are compact and they are disjoint, thus we have~$R_i>0$ well-defined.

We now bound~$\psi_{2, \epsilon}$ in the regions~$B_1$ and~$B_2$, respectively.

\begin{lemma}[Lemma 8.2 in \cite{han2025polynomialcomplexitysamplingmultimodal}]\label{lem: boundedness of eigenfunction around local minima}
  There exists  a constant~$C_a= C_a(d,U,C_m,\alpha)$ and~$\tilde{\epsilon}=\tilde{\epsilon}(d,\hat{\gamma},\alpha)$ such that for every
  \begin{equation}\label{eq:epsilonCri}
      0 < \epsilon\leq \min\set[\Big]{R_1,R_2,\tilde{\epsilon}}\defeq \epsilon_1
  \end{equation}
  we have
\begin{equation}\label{eq: eigenvector in Bi}
    \forall x\in B_i, \quad  |\psi_{2,\epsilon}(x)-a_{i,\epsilon}|\leq C_a
      \exp\paren[\Big]{-\frac{\hat{\gamma}}{4(1+\alpha)^{\frac{1}{8}} \epsilon}}, \quad i=1,2.
\end{equation}
 Here  ~$a_{1,\epsilon}$ and~$a_{2,\epsilon}$ are defined as in \eqref{eq: ak sol}. 
\end{lemma}

\begin{proof}[Proof of Lemma~\ref{lem: boundedness of eigenfunction around local minima}]

Fix~$i=1$ or~$2$,
for each~$x\in B_i$, there exists~$y\in B_i$ such that~$x\in B(y,\epsilon)$.
  By the triangle inequality, it follows that~$B(y,2\epsilon)\subseteq \widetilde{B}_i$. Thus,
\begin{equation}\label{eq:concirc}
    B(y,2\epsilon)\subseteq \widetilde{B}_i\subset\Omega_i.
\end{equation}

First notice that the function~$\psi_{2,\epsilon}-a_{i,\epsilon}$ satisfies
\[
(L_{\epsilon}-\lambda_{2,\epsilon})(\psi_{2,\epsilon}-a_{i,\epsilon})=\lambda_{2,\epsilon}a_{i,\epsilon}.
\]
Thus using~\cite[Corollary 9.21]{GT}, there exists dimensional constant~$C$ such that for every~$y\in B_i$ for which~$B(y,2\epsilon)\subseteq \widetilde{B}_i$, we have
\begin{align}
    \sup_{x\in B(y,\epsilon)} \abs{\psi_{2,\epsilon}(x)-a_{i,\epsilon}}
	\leq C\biggl(&\paren[\Big]{\frac{1}{|B(y,2\epsilon)|}\int_{B(y,2\epsilon)}|\psi_{2,\epsilon}(x)-a_{i,\epsilon}|^2 \, d x}^{\frac12}
    \\\label{eq:GTlocal}
	& +|\lambda_{2,\epsilon}a_{i,\epsilon}|\biggr).
\end{align}

We first bound the term~$|\lambda_{2,\epsilon}a_{i,\epsilon}|$. Observe that according to \eqref{eq: egvalLEpsilon}, there exists constant~$\tilde{C}\equiv \tilde{C}(U)$ such that for every~$\epsilon<1$,
\begin{equation}
    \lambda_{2,\epsilon}\leq \tilde{C}\exp\paren[\Big]{\frac{\hat{\gamma}}{4\epsilon}}.
\end{equation}
Meanwhile, according to \eqref{eq: ak sol} and Assumption~\ref{a:massRatioBound}, we have that~$|a_{i,\epsilon}|\leq C_m$. Therefore, there exists constant~$C\equiv C(C_m,U)$ such that
\begin{equation}\label{eq:psiterm2}
    |\lambda_{2,\epsilon}a_{i,\epsilon}|\leq C\exp\paren[\Big]{\frac{\hat{\gamma}}{4\epsilon}}.
\end{equation}

Now we bound~$\int_{B(y,2\epsilon)}|\psi_{2,\epsilon}(x)-a_{i,\epsilon}|^2 \, d x$ that appears on the right hand side of~\eqref{eq:GTlocal}. Using the fact that when \eqref{eq:epsilonCri} holds, for~$i=1,2$,
\begin{align}
    Z_{\epsilon}&=\Big(\int_{\Omega_i}e^{-\frac{U}{\epsilon}}\, d x\Big)\cdot\Big( 1+\frac{\int_{\mathbb{R}^d\setminus\Omega_i}e^{-\frac{U}{\epsilon}}\, d x}{\int_{\Omega_i}e^{-\frac{U}{\epsilon}}\, d x}\Big)\\
     &=\Big(\int_{\Omega_i}e^{-\frac{U}{\epsilon}}\, d x\Big)\cdot\Big( 1+\frac{1-\pi_{\epsilon}(\Omega_i)}{\pi_{\epsilon}(\Omega_i)}\Big)\overset{\eqref{e:massRatioBound}}{\leq} \Big(\int_{\Omega_i}e^{-\frac{U}{\epsilon}}\, d x\Big)\cdot( 1+C_m^2)\\
    \label{eq:intoi} &\overset{\mathclap{\eqref{eq:Zepsilon}}}{\leq} C(2\pi\epsilon)^{\frac{d}{2}} e^{-\frac{U(x_{\min,i})}{\epsilon}},
\end{align}
we have that 
\begin{align}
   \MoveEqLeft \int_{B(y,2\epsilon)}|\psi_{2,\epsilon}-a_{i,\epsilon}|^2 \, d x= \int_{B(y,2\epsilon)}|\psi_{2,\epsilon}-a_{i,\epsilon}\one_{\Omega_i}|^2 \, d x \\
        &\leq  \Big(\sup_{z\in B(y,2\epsilon)}e^{\frac{U(z)}{\epsilon}}\Big)\int_{B(y,2\epsilon)}|\psi_{2,\epsilon}-a_{i,\epsilon}\one_{\Omega_i}|^2 e^{-\frac{U}{\epsilon}}\, d x\\
        &\overset{\mathclap{\eqref{eq:concirc}}}{\leq}Z_{\epsilon}\Big(\sup_{z\in \widetilde{B}_i}e^{\frac{U(z)}{\epsilon}}\Big)\int_{\Omega_i}|\psi_{2,\epsilon}(x)-a_{i,\epsilon}\one_{\Omega_i}|^2 \, d \pi_{\epsilon}(x)\\
       &\overset{\mathclap{\eqref{eq:intoi}}}{\leq} C(2\pi\epsilon)^{\frac{d}{2}}\Big(\sup_{z\in \widetilde{B}_i}e^{\frac{U(z)-U(x_{\min,i})}{\epsilon}}\Big)\norm[\big]{\psi_{2,\epsilon}-a_{1,\epsilon}\one_{\Omega_1}-a_{2,\epsilon}\one_{\Omega_2}}_{L^2(\pi_{\epsilon})}^2\\
       &\overset{\mathclap{\eqref{eqVkLessThanDEkFk}}}{\leq}C(2\pi\epsilon)^{\frac{d}{2}}\Big(\sup_{z\in \widetilde{B}_i}e^{\frac{U(z)-U(x_{\min,i})}{\epsilon}}\Big)d(E_{\epsilon},F_{\epsilon})^2\\
       &\overset{\mathclap{\eqref{e:dEF}}}{\leq}C(2\pi\epsilon)^{\frac{d}{2}}\Big(\sup_{z\in \widetilde{B}_i}e^{\frac{U(z)-U(x_{\min,i})}{\epsilon}}\Big)\exp\paren[\Big]{-\frac{2\hat{\gamma}}{(1+\alpha)^{\frac18}\epsilon}}\\
       &\overset{\mathclap{\eqref{eq: deftildeB}}}{\leq}C(2\pi\epsilon)^{\frac{d}{2}}\exp\paren[\Big]{-\frac{\hat{\gamma}}{ (1+\alpha)^{\frac18}\epsilon}},  \label{eq:13over8hatgamma}
\end{align}
where the second last inequality we use \eqref{e:dEF} with~$\gamma =\frac{\hat{\gamma}}{(1+\alpha)^{\frac18}}$.

Notice that there exists constant~$\tilde{\epsilon}=\tilde{\epsilon}(d,\hat\gamma,\alpha)$ that whenever~$\epsilon<\tilde{\epsilon}$, 
\begin{equation}\label{eq:tildeEps}
    \exp\paren[\Big]{-\frac{\hat{\gamma}}{2(1+\alpha)^{\frac{1}{8}} \epsilon}}< (2\pi\epsilon)^{\frac{d}{2}}.
\end{equation}
Thus, for~$\epsilon<\tilde{\epsilon}$,
\begin{equation} \label{eq:nearBi}
    \int_{B(y,2\epsilon)}|\psi_{2,\epsilon}-a_{i,\epsilon}|^2 \, d x  \overset{\eqref{eq:13over8hatgamma},\eqref{eq:tildeEps}}{\leq} C(2\pi\epsilon)^{d}\exp\paren[\Big]{-\frac{\hat{\gamma}}{2(1+\alpha)^{\frac{1}{8}} \epsilon}}. 
\end{equation}

Therefore, plugging \eqref{eq:nearBi} into \eqref{eq:GTlocal} gives
\begin{align*}
     \MoveEqLeft \sup_{x\in B(y,\epsilon)} |\psi_{2,\epsilon}(x)-a_{1,\epsilon}|
	~\overset{\mathclap{\eqref{eq:nearBi}}}{\leq}~
	    \paren[\Big]{
	      \frac{C(2\pi\epsilon)^{d}}{|B(y,2\epsilon)|}\exp\paren[\Big]{-\frac{\hat{\gamma}}{2(1+\alpha)^{\frac{1}{8}} \epsilon}}
	    }^{\frac12}
	  +C|\lambda_{2,\epsilon}a_{1,\epsilon}|
      \\
       &\overset{\mathclap{\eqref{eq:psiterm2}}}{\leq}\quad
	  \paren[\Big]{\frac{C(2\pi\epsilon)^{d}}{( 2\epsilon)^d}\exp\paren[\Big]{-\frac{\hat{\gamma}}{2(1+\alpha)^{\frac{1}{8}} \epsilon}}}^{\frac12}
        +C\exp\paren[\Big]{-\frac{\hat{\gamma}}{4\epsilon}}
	\\
        &\leq  C_a\exp\paren[\Big]{-\frac{\hat{\gamma}}{4(1+\alpha)^{\frac{1}{8}} \epsilon}},
\end{align*}
which implies \eqref{eq: eigenvector in Bi}.
\end{proof}

We now use the above lemma and Assumption~\ref{a:massRatioBound} to prove Lemma~\ref{lem:boundPsi} .

\begin{proof}[Proof of Lemma~\ref{lem:boundPsi}]

We discuss two cases,~$\epsilon\leq \epsilon_1$ and~$\epsilon> \epsilon_1$, where~$\epsilon_1$ is defined in \eqref{eq:epsilonCri}. 

\restartcases

\case[$\epsilon\leq \epsilon_1$] We obtain immediately from \eqref{eq: eigenvector in Bi} that
\begin{equation}
    \sup\limits_{0< \epsilon\leq \epsilon_1}\|\psi_{2,\epsilon}\|_{L^{\infty}(K)}\leq C_m+C_a.
\end{equation}

\case[$1\geq \epsilon> \epsilon_1$] Define the set~$K_2$ as the~$2$-neighborhood of~$K$. Namely,
\begin{equation}
    K_2 \defeq \overline{\bigcup_{y\in K} B(y,2)}.
\end{equation}
It is straightforward to see that~$K_2\subset\mathbb{R}^d$ is bounded and closed.
According to \cite[Corollary 9.21]{GT}, for~$y\in K$,
\begin{align*}
    \sup_{x\in B(y,\epsilon)} |\psi_{2,\epsilon}(x)|
        &\leq \paren[\Big]{\frac{C}{|B(y,2\epsilon)|}\int_{B(y,2\epsilon)}|\psi_{2,\epsilon}(x)|^2 \, d x}^{\frac12}
       \\
        &\leq \paren[\Big]{\frac{C}{|B(y,2\epsilon)|}\Big(\sup_{z\in K_2}e^{\frac{U(z)-U_{\min}}{\epsilon}}\Big)\int_{\mathbb{R}^d}|\psi_{2,\epsilon}(x)|^2 \, d \pi_{\epsilon}(x)}^{\frac12}\\
        &=C(\epsilon_1)^{-\frac{d}{2}}\exp\paren[\Big]{\frac{\norm{U}_{L^{\infty}(K_2)}}{2\epsilon_1}}= C(U,d, C_m,\alpha).
\end{align*}

We conclude from the above two cases that~\eqref{eq:defCpsi} holds.
\end{proof}

\section{Iterating error estimates (Lemma~\ref{l:iteration})} \label{sec:iteration}

Lemma~\ref{l:iteration} consists of three main parts: the derivation of recurrence relation~\ref{e:iteration}, control the probability that \eqref{eq:Xk+1inK} and \eqref{e:iteration} hold, and obtaining the estimate \eqref{e:CBeta} for~$\beta_k$ and~$c_k$. We do each of these steps in Sections~\ref{sec: recurrence},~\ref{sec:estCj} and~\ref{sec: est coeff}.
We combine these and prove Lemma~\ref{l:iteration} in Section~\ref{s:proofIteration}.

\subsection{Recurrence relation}\label{sec: recurrence}

We will now prove \eqref{e:iteration} by combining the estimate
for the Monte Carlo error (Lemma~\ref{l:langevinError}) and the resampling error (Lemma~\ref{l:rebalancing}). For
clarity, we state this as a new lemma and give explicit formulae for the constants~$\beta_k$
and~$c_k$ appearing in \eqref{e:iteration}.

\begin{lemma}\label{lem: iterative scheme between levels}
  Let~$2\leq k\leq M-1$. Assume that~$X_{k,0}^{i}\in K$ for all~$i\in\set{1,\dots, N}$. Then with probability larger than or equal to
   \begin{align}
  \theta_k\defeq 1&-N\paren[\bigg]{C_P(C_{\psi}+1)\exp\paren[\Big]{-\frac{C_{K}}{\eta_k}}-e^{-\Lambda T/2}\max_{x\in K}\sqrt{\frac{p_{k,T}(x,x)}{\pi_{k}(x)}-1}}\\
   &-2\exp\paren[\Big]{-\frac{2Na_1^2}{C_r^2}}-2\exp\paren[\Big]{-\frac{Na_2^2}{2C^2_{r}C_{\psi}^2}}-2\exp\paren[\Big]{-\frac{Na_3^2}{2C_{\psi}^2}}
   \label{e:probIteration}
\end{align}
   we have that \eqref{eq:Xk+1inK} and \eqref{e:iteration} hold with
  \begin{align}
    \beta_k &\defeq e^{-\lambda_{2, k} T}\Big(\abs[\Big]{\int \psi_{2,k}\pi_{k+1} \, dx}\cdot\norm[\Big]{\psi_{2,k+1}\one_{K}-\int_{K}\psi_{2,k+1}\pi_{k+1}\, d x}_{L^\infty} 
    \\
      &\qquad\qquad\qquad \mathbin{+}\abs[\Big]{\int \paren[\big]{\psi_{2,k+1}\one_{K}-\int_{K}\psi_{2,k+1}\pi_{k+1}\, d x}\psi_{2,k}\pi_{k+1}\, dx}\Big) \label{eq: defBetak}
    \\
    c_k &\defeq \norm[\Big]{\psi_{2,k+1}\one_{K}-\int_{K}\psi_{2,k+1}\pi_{k+1}\, d x}_{L^\infty} \Big(\mathcal E_{k,T}(r_k)  +a_1\Big)\\
    &\qquad +\mathcal E_{k,T}(\psi_{2,k+1}r_k\one_{K}) +a_2+a_3+ \abs[\Big]{\int_{K}\psi_{2,k+1}\pi_{k+1}\, d x}
    .  \label{eq: defck}
\end{align}
\end{lemma}
\begin{proof}[Proof of Lemma~\ref{lem: iterative scheme between levels}]

The proof consists of two steps, estimating the error in the  Langevin and the resampling step, respectively.
\restartsteps

\step[Langevin step]
We apply Lemma~\ref{l:langevinError} with
\[
\epsilon=\eta_k, \quad h=r_k,\quad q_{\epsilon,0}=q_{k,0}
,
\]
to obtain that with probability larger than or equal to 
\begin{equation}\label{eq:LprobR}
    1-2\exp\paren[\Big]{-\frac{2Na_1^2}{C_r^2}}
\end{equation}
  we have
  \begin{align}
    \abs[\Big]{ \frac{1}{N}\sum_{i=1}^{N}r_k(X^{i}_{k, T})-1}
    &\leq e^{-\lambda_{2, k} T} \abs[\Big]{\int \psi_{2,k}\pi_{k+1} \, dx} \cdot\abs[\Big]{\frac{1}{N}\sum_{i=1}^{N}\psi_{2,k}\one_{K}(X^{i}_{k,0})}
    \\
    &
    + \mathcal E_{k,T}(r_k)  +a_1. \label{eq: ErrR}
\end{align}

Similarly, we apply Lemma~\ref{l:langevinError} with
\[
\epsilon=\eta_k, \quad h=\paren{\psi_{2,k+1}\one_{K}}r_k-\int_{K}\psi_{2,k+1}\pi_{k+1}\, d x,\quad q_{\epsilon,0}=q_{k,0}
\] 
to obtain that with probability greater than or equal to 
\begin{equation}\label{eq:LprobpsiR}
    1-2\exp\paren[\Big]{-\frac{Na_2^2}{2C_r^2C_{\psi}^2}}
\end{equation}
we have
\begin{multline}
 \abs[\Big]{\frac{1}{N} \sum_1^N \paren[\Big]{\psi_{2,k+1}\one_{K}(X_{k,T}^{i})-\int_{K}\psi_{2,k+1}\pi_{k+1}\, d x}r^k(X_{k,T}^{i})}\\
\overset{\eqref{e:langevinError}}{\leq}e^{-\lambda_{2, k} T}\abs[\Big]{\int \paren[\big]{\psi_{2,k+1}\one_{K}-\int_{K}\psi_{2,k+1}\pi_{k+1}\, d x}\psi_{2,k}\pi_{k+1}\, dx}\cdot\abs[\Big]{\frac{1}{N}\sum_{i=1}^{N}\psi_{2,k}\one_{K}(X^{i}_{k,0})}\\
	+\mathcal E_{k,T}(\psi_{2,k+1}r_k\one_{K}) +a_2\label{eq: ErrPsiR}
\end{multline}
where we use the fact that
\begin{equation}
    \norm{r_k\psi_{2,k+1}\one_{K}}_{\osc}\leq 2 \norm{r_k\psi_{2,k+1}\one_{K}}_{L^{\infty}}\leq 2\norm{r_k}_{L^{\infty}}\norm{\psi_{2,k+1}\one_{K}}_{L^{\infty}}\leq 2C_rC_{\psi}.
\end{equation}
 Combining \eqref{eq:LprobR}, \eqref{eq:LprobpsiR} and \eqref{eq:lbxinK},  
after the Langevin step, with probability larger than or equal
\begin{align}
   1&-N\paren[\bigg]{C_P(C_{\psi}+1)\exp\paren[\Big]{-\frac{C_{K}}{\eta_k}}+e^{-\Lambda T/2}\max_{x\in K}\sqrt{\frac{p_{k,T}(x,x)}{\pi_{k}(x)}-1}}\\
   &-2\exp\paren[\Big]{-\frac{2Na_1^2}{C_r^2}}-2\exp\paren[\Big]{-\frac{2Na_2^2}{C_{\psi}^2}} \label{eq:probI1}
\end{align}
we have \eqref{eq: ErrR} \eqref{eq: ErrPsiR} hold as well as
\begin{equation}\label{eq:XkinK}
    X_{k,T}^{i}\in K,\quad\forall i.
\end{equation}
\step[Resampling step]
Observe that \eqref{eq:XkinK} guarantees that \eqref{eq:Xk+1inK}.
  Applying Lemma~\ref{l:rebalancing} with
\[
p=\pi_{k},\quad q=\pi_{k+1},\quad h=\psi_{2,k+1}\one_{K},\quad
x^{i}=X_{k,T}^{i},\quad y^{i}=X^{i}_{k+1,0}
\]
 gives that with probability larger than or equal to
 \begin{equation}
     1- 2\exp\paren[\Big]{-\frac{Na_3^2}{2C_{\psi}^2}} \label{eq:probI2}
 \end{equation}
 we have
\begin{align}
  \MoveEqLeft
    \abs[\Big]{\frac{1}{N}\sum_{i=1}^{N}\psi_{2,k+1}\one_{K}(X_{k+1,0}^i)-\int_{K}\psi_{2,k+1}\pi_{k+1}\, d x}
      \\
      &\overset{\eqref{e:resampling}}{\leq}
	 \norm[\Big]{\psi_{2,k+1}\one_{K}-\int_{K}\psi_{2,k+1}\pi_{k+1}\, d x}_{L^\infty}  \abs[\Big]{ \frac{1}{N}\sum_{i=1}^{N}r_k(X^{i}_{k, T})-1}\\
     &+ \abs[\Big]{\frac{1}{N} \sum_1^N \paren[\Big]{\psi_{2,k+1}\one_{K}(X_{k,T}^{i})-\int_{K}\psi_{2,k+1}\pi_{k+1}\, d x}r^k(X_{k,T}^{i})}+a_3
    \label{eq: ErrK+1Decomp}
    \end{align}
where we use the fact that
\begin{equation}
    \norm{\psi_{2,k+1}\one_{K}}_{\osc}\leq 2 \norm{\psi_{2,k+1}\one_{K}}_{L^{\infty}}\leq 2\norm{\psi_{2,k+1}\one_{K}}_{L^{\infty}}\leq 2C_{\psi}.
\end{equation}
The probability \eqref{e:probIteration} is obtained by multiplying \eqref{eq:probI1} and \eqref{eq:probI2} then using the inequality~$(1-a)(1-b)>1-a-b$ for~$a,b>0$. 
 Plugging~\eqref{eq: ErrR} and~\eqref{eq: ErrPsiR} into~\eqref{eq: ErrK+1Decomp} and using~\eqref{eq: defBetak},~\eqref{eq: defck} yields~\eqref{e:iteration}, completing the proof.
\end{proof}

\subsection{Estimate of \texorpdfstring{$c_k$}{ck}} \label{sec:estCj}

The estimate of~$c_k$ is essentially a straightforward calculation. We first choose proper~$a_i$,~$i=1,2,3$ and then show how~$N$ and~$T$ can be chosen so that we obtain the bound for~$c_k$ in~\eqref{e:CBeta}. Precisely, we have the following lemma.

\begin{lemma}[Estimate of~$c_k$]\label{lem: cj}
  Fix~$\delta>0$, let~$C_{\psi}$,~$C_{r}$ be the constants defined in \eqref{eq:defCpsi} and \eqref{eq:defCr} respectively. Let 
\begin{equation}\label{eq:achoose}
    a_1\defeq \frac{\delta}{8M(C_{\psi}+1)}, \quad a_2\defeq \frac{\delta}{8M}, \quad a_3\defeq \frac{\delta}{8M}.
\end{equation}
  If~$T$ is chosen such that
  \begin{equation}
     T\geq \max\Big\{\frac{2}{\Lambda}\Big(\log\paren[\Big]{\frac{M}{\delta}}+ \frac{\hat{U}}{2\eta}+\frac{1}{8}+\log(4C_p^{\frac12}C^{\frac12}_{r}(C_{\psi}+1))\Big), 
\frac{c}{d}\log (4d)\Big\}, \label{eq:Tck}
  \end{equation}
then for every~$\eta_k$ such that
\begin{equation}\label{e:etaCriC}
    \eta_k\leq \frac{C_K}{\log\paren[\Big]{\frac{8MC_P}{\delta}}}
\end{equation}
we have
\begin{equation*}
  c_k\leq  \frac{\delta}{M}.
\end{equation*}
\end{lemma}

\begin{proof}[Proof of Lemma~\ref{lem: cj}]
The bound of~$c_k$ requires several a priori estimates.
\restartsteps
\step[Bound~$\mathcal E_{k,T}$] When~$T$ satisfies \eqref{eq:Tck}, we have
\begin{equation}
    e^{-\Lambda T/2}\leq \frac{\delta}{M}\exp\paren[\Big]{\frac{\hat{U}}{2\eta_k}}  e^{-\frac{1}{8}}\frac{1}{4C_p^{\frac12}C_r^{\frac12}(C_{\psi}+1)}
\end{equation}
which implies that
\begin{equation}\label{eq:eLambdaTmaxpoverpi}
    e^{-\Lambda T/2}\max_{x\in K}\sqrt{\frac{p_{k,T}(x,x)}{\pi(x)}-1}\overset{\eqref{eq:poverpiK},\eqref{eq:ubPatTcri}}{\leq} \frac{\delta}{4MC_r^{\frac12}(C_{\psi}+1)}. 
\end{equation}
Observe that that 
\begin{equation}
    \norm{\psi_{2,k+1}r_k\one_{K}}_{L^2(\pi_k)}\leq \norm{\psi_{2,k+1}r_k}_{L^2(\pi_k)}\leq \norm{\psi_{2,k+1}}_{L^2(\pi_{k+1})}\norm{r_k}^{\frac12}_{L^{\infty}(\pi_k)}\overset{\eqref{eq:defCr}}{\leq} C_r^{\frac12}
\end{equation}
which implies that
\begin{equation}
    \mathcal E_{k,T}(\psi_{2,k+1}r_k\one_{K})\overset{\eqref{eq: def of tilde epsilon}}{\leq}C_r^{\frac12}e^{-\Lambda T/2}\max_{x\in K}\sqrt{\frac{p_{\epsilon,T}(x,x)}{\pi(x)}-1} \overset{\eqref{eq:eLambdaTmaxpoverpi}}{\leq}\frac{\delta}{4M(C_{\psi}+1)}. \label{eq:ckT3}
\end{equation}
Similarly, the fact that 
\begin{equation}
    \norm{r_k\one_{K}}_{L^2(\pi_k)}\leq \norm{r_k}^{\frac12}_{L^{\infty}(\pi_k)}\overset{\eqref{eq:defCr}}{\leq} C_r^{\frac12}
\end{equation}
implies 
\begin{equation}\label{eq:ckT2}
    \mathcal E_{k,T}(r_k)\leq \frac{\delta}{4M(C_{\psi}+1)}.
\end{equation}
\step
Observe that
\begin{align}
  \norm[\Big]{\psi_{2,k+1}\one_{K}-\int_{K}\psi_{2,k+1}\pi_{k+1}\, d x}_{L^\infty}
   &\leq \norm{\psi_{2,k+1}}_{L^{\infty}(K)}+\abs[\Big]{\int_{K}\psi_{2,k+1}\pi_{k+1}\, d x}\\
   &\overset{\mathclap{\eqref{eq:defCpsi}}}{\leq}C_{\psi}+\norm{\psi_{2,k+1}}_{L^2(\pi_{k+1})}=C_{\psi}+1. \label{eq:ckT1}
\end{align}
And when~$\eta_k$ satisfies \eqref{e:etaCriC},
\begin{equation}\label{eq:ckT5}
   \abs[\Big]{\int_{K}\psi_{2,k+1}\pi_{k+1}\, d x}\overset{\eqref{e:intKpsipi}}{\leq} C_P\exp\paren{-\frac{C_K}{\eta_{k+1}}}\leq \frac{\delta}{8M}.
\end{equation}
Therefore, when \eqref{eq:achoose}, \eqref{eq:Tck} and \eqref{e:etaCriC} holds,
\begin{align}
    c_k\quad&\overset{\mathclap{\eqref{eq: defck},\eqref{eq:ckT1},\eqref{eq:ckT5}}}{\leq} \quad (C_{\psi}+1)\cdot\Big(\mathcal E_{k,T}(r_k)  +a_1\Big)+\mathcal E_{k,T}(\psi_{2,k+1}r_k\one_{K}) +a_2+a_3+ \frac{\delta}{8M}\\
    &\overset{\mathclap{\eqref{eq:achoose},\eqref{eq:ckT3},\eqref{eq:ckT2}}}{\leq} \qquad \quad \frac{3\delta}{8M}+\frac{\delta}{4M}+\frac{\delta}{8M}+\frac{\delta}{8M}+\frac{\delta}{8M}=\frac{\delta}{M}. \qedhere
\end{align}

\end{proof}

\subsection{Estimate of \texorpdfstring{$\beta_k$}{betak}} \label{sec: est coeff}

Recall from~\eqref{e:iteration}, the error grows by a factor of~$\beta_k$ at each level, and so to prove Theorem~\ref{thm: main} we need to ensure~$\prod \beta_k$ remains bounded.
The main result in this section (Lemma~\ref{lem: betaj}, below) obtains this bound and shows that the first inequality in~\eqref{e:CBeta} holds.
For simplicity of notation, let
\begin{multline}\label{eq:defTheta}
    \Theta(k,k+1)\defeq \abs[\Big]{\int \psi_{2,k}\pi_{k+1} \, dx}\cdot\norm[\Big]{\psi_{2,k+1}\one_{K}-\int_{K}\psi_{2,k+1}\pi_{k+1}\, d x}_{L^\infty} 
    \\
 \mathbin{+}\abs[\Big]{\int \paren[\big]{\psi_{2,k+1}\one_{K}-\int_{K}\psi_{2,k+1}\pi_{k+1}\, d x}\psi_{2,k}\pi_{k+1}\, dx} ,
\end{multline}
and note
\[
\beta_k=e^{-\lambda_{2,k}T}\Theta(k,k+1).
\]
When we bound~$\prod_{j=k}^{M-1}\beta_j$ in the low temperature regime, the exponential factor~$e^{-\lambda_{2,k}T}$ is very close to~$1$, and does not help much.
Thus we show that the product~$\prod_{j=k}^{M-1}\Theta(j,j+1)$ stays bounded, by approximating~$\Theta(k,k+1)$ in terms of the mass in each well and estimating the mass distribution using small temperature asymptotics.

We use Proposition~\ref{prop_kolo_chap8_prop2.2} to estimate the two integration terms appearing in~$\Theta(k,k+1)$.
The bounds we need are stated in the next two lemmas, and their proofs are identical to those of \cite[lemma 7.5, Lemma 7.6]{han2025polynomialcomplexitysamplingmultimodal}. We omit their proof here.
\begin{lemma}[Lemma 7.5 in \cite{han2025polynomialcomplexitysamplingmultimodal}]\label{lem: integral of two eigenfunction from different levels}
  
Let~$\epsilon'<\epsilon$ and define~$r_\epsilon$ by
  \begin{equation}\label{e:rdefEps}
    r_{\epsilon}\defeq\frac{\pi_{\epsilon'}}{\pi_{\epsilon}}
    \,.
  \end{equation}
  Then,
\begin{align}
    \norm{\psi_{2,\epsilon}}_{L^2(\pi_{\epsilon'})}
     &\leq \|r_{\epsilon}\|^{\frac12}_{L^{\infty}}d(E_{\epsilon},F_{\epsilon})\\
  \label{eq:intPsiPsiPi} 
    &\qquad +\paren[\bigg]{1+\paren[\Big]{\frac{\pi_{\epsilon}(\Omega_2)}{\pi_{\epsilon}(\Omega_1)}-\frac{\pi_{\epsilon}(\Omega_1)}{\pi_{\epsilon}(\Omega_2)}}(\pi_{\epsilon'}(\Omega_1)-\pi_{\epsilon}(\Omega_1))}^{\frac12}. 
\end{align}
\end{lemma}

\begin{lemma}[Lemma 7.6 in \cite{han2025polynomialcomplexitysamplingmultimodal}]\label{lem: integral of eigenfunction on different measure}
    
    Let~$\epsilon'<\epsilon$. Then
    \begin{align}
      \MoveEqLeft
      \Big|\int \psi_{2,\epsilon}\pi_{\epsilon'}\, d x\Big|\\
       &\leq  \bigg( \frac{\sqrt{\pi_{\epsilon}(\Omega_2)}}{\sqrt{\pi_{\epsilon}(\Omega_1)}}+\frac{\sqrt{\pi_{\epsilon}(\Omega_1)}}{\sqrt{\pi_{\epsilon}(\Omega_2)}}\bigg)\cdot |\pi_{\epsilon'}(\Omega_1)-\pi_{\epsilon}(\Omega_1)|+d(E_{\epsilon},F_{\epsilon})\|r_{\epsilon}\|_{L^{\infty}(\pi_{\epsilon})}.  \label{eq:intPsiPi}
    \end{align}
\end{lemma}

To apply the previous two results we need to ensure the masses in the two wells stay away from~$0$ (Assumption~\ref{a:massRatioBound}), and do not oscillate too much. The required oscillation condition holds provided~$U$ satisfies Assumption~\ref{a:criticalpts} holds, as the following lemma states. 

\begin{lemma}[Lemma 7.7 in \cite{han2025polynomialcomplexitysamplingmultimodal}]\label{lem: mass in well as BV}
  If~$U$ satisfies Assumption~\ref{a:criticalpts}, then there exists a constant~$C_{\BV}$ such that such that for every~$\eta\in (0,1)$, and every~$i\in \set{1,2}$ we have
\begin{equation}\label{eq:massBV}
	 \int_{\eta}^{1}\abs{\partial_{\epsilon}\pi_{\epsilon}(\Omega_i)}\, d \epsilon\leq C_{\BV}.
       \end{equation}
\end{lemma}
The proof of Lemma~\ref{lem: mass in well as BV} is identical to that of Lemma~7.7 in \cite{han2025polynomialcomplexitysamplingmultimodal} and is thus omitted here. We now bound~$\prod \beta_j$ to obtain the first inequality in~\eqref{e:CBeta}.

\begin{lemma}[Estimate of~$\beta_j$]\label{lem: betaj}

Define~$k_0$ by
\begin{equation}\label{eq:defEtacr}
    k_0\defeq\min\Bigl\{2\leq k\leq M-1 \: | \: \eta_k\leq \frac{C_K}{\log\paren{{C}_{\gamma}M}}\Bigr\},
\end{equation}
where~$C_\gamma$ is the constant defined in \eqref{e:dEF} when~$\gamma\defeq C_K$.
If at each step~$T>0$,
then for~$k_0\leq k\leq M-1$, the inequality \eqref{e:CBeta}
  holds with
\begin{equation}\label{eq: defCbeta}
  C_{\beta}\defeq \exp\paren[\Big]{C_{\BV}\paren[\Big]{2C_m (C_{\psi}+2)+C_m^2}+ (C_{\psi}+2)C_r+C_r^{\frac12}}.
\end{equation}
  Here~$C_{r}$,~$C_{m}$,~$C_{\psi}$ and~$C_{\BV}$ are the constants defined in
~\eqref{eq:defCr},
~\eqref{e:massRatioBound},
~\eqref{eq:defCpsi},
  and~\eqref{eq:massBV},
  respectively.
\end{lemma}

\begin{proof}[Proof of Lemma~\ref{lem: betaj}]

Observe that by Proposition~\ref{prop_kolo_chap8_prop2.2}, we have
\begin{equation}\label{eq:dleq1overM}
d(E_{\eta_k},F_{\eta_k})
   \overset{\eqref{e:dEF}}{\leq} C_{\gamma}\exp\paren[\Big]{-\frac{C_K}{\eta_k}}\overset{\eqref{eq:defEtacr}}{\leq} \frac{1}{M}.
\end{equation}
To bound~$\Theta(k,k+1)$, we write
\begin{equation}\label{eq:thetarewrite}
    \Theta(k,k+1)\overset{\eqref{eq:defTheta}}{=} J_1\norm[\Big]{\psi_{2,k+1}\one_{K}-\int_{K}\psi_{2,k+1}\pi_{k+1}\, d x}_{L^\infty}+J_2,
\end{equation}
where 
\begin{equation}\label{eq:thetaj1j2}
    J_1=\Big|\int \psi_{2,k}\pi_{k+1}\, dx\Big|,\quad
    J_2=\abs[\Big]{\int \paren[\big]{\psi_{2,k+1}\one_{K}-\int_{K}\psi_{2,k+1}\pi_{k+1}\, d x}\psi_{2,k}\pi_{k+1}\, dx}.
\end{equation}

\restartsteps
\step[Estimating~$J_1$ and~$J_2$] We first estimate~$J_1$ and~$J_2$ using Lemma~\ref{lem: integral of two eigenfunction from different levels} and~\ref{lem: integral of eigenfunction on different measure} respectively.
For simplicity, by a slight abuse of notation we write
\begin{equation}
    \pi_k(\Omega_i)\defeq \pi_{\eta_k}(\Omega_i),\quad i=1,2.
\end{equation}

By Lemma~\ref{lem: integral of eigenfunction on different measure},
\begin{align}
   J_1&\overset{\mathclap{\eqref{eq:intPsiPi}}}{\leq} ~d(E_{\eta_k},F_{\eta_k})\norm{r_k}_{L^{\infty}}+
       \Big( \frac{\sqrt{\pi_{k}(\Omega_2)}}{\sqrt{\pi_{k}(\Omega_1)}}+\frac{\sqrt{\pi_{k}(\Omega_1)}}{\sqrt{\pi_{k}(\Omega_2)}}\Big)\cdot |\pi_{k+1}(\Omega_1)-\pi_{k}(\Omega_1)|\\
       &\overset{\mathclap{\eqref{e:massRatioBound}}}{\leq}~ d(E_{\eta_k},F_{\eta_k})\norm{r_k}_{L^{\infty}}
       +2C_m\cdot \big|\pi_{k+1}(\Omega_1)-\pi_{k}(\Omega_1)\big|\\
       \label{eq:thetaj1}&\overset{\mathclap{\eqref{eq:dleq1overM}}}{\leq}~\frac{C_r}{M}+2C_m\cdot \big|\pi_{k+1}(\Omega_1)-\pi_{k}(\Omega_1)\big|.
\end{align}

For term~$J_2$, observe that
\begin{equation}
    J_2\leq J_3+J_4
\end{equation}
where
\begin{align}
    J_3&\defeq \abs[\Big]{\int \paren[\big]{\psi_{2,k+1}\one_{K}}\psi_{2,k}\pi_{k+1}\, dx}\\
    J_4&\defeq \abs[\Big]{ \paren[\Big]{\int_{K}\psi_{2,k+1}\pi_{k+1}\, d x}\paren[\Big]{\int \psi_{2,k}\pi_{k+1}\, dx}}
\end{align}
By Lemma~\ref{lem: integral of two eigenfunction from different levels}, using the fact that~$(1+y)^{\frac12}\leq 1+\frac12 y$ when~$y>0$, we have
\begin{align}
    J_3&\leq \paren[\Big]{\int_{K}\abs{\psi_{2,k+1}}^2\pi_{k+1}\,d x}^{\frac12}\paren[\Big]{\int_{K}\abs{\psi_{2,k}}^2\pi_{k+1}\,d x}^{\frac12}\leq \norm{\psi_{2,k}}_{L^2(\pi_{k+1})}\\
    &\overset{\mathclap{\eqref{eq:intPsiPsiPi}}}{\leq} \|r_k\|^{\frac12}_{L^{\infty}}d(E_{\eta_k},F_{\eta_k})
      +\Big(1+\paren[\Big]{\frac{\pi_{k}(\Omega_2)}{\pi_{k}(\Omega_1)}-\frac{\pi_{k}(\Omega_1)}{\pi_{k}(\Omega_2)}}(\pi_{k+1}(\Omega_1)-\pi_{k}(\Omega_1))\Big)^{\frac12}\\
      &\overset{\mathclap{\eqref{eq:dleq1overM},\eqref{eq:defCr},\eqref{e:massRatioBound}}}{\leq} ~\qquad \frac{C_r^{\frac12}}{M}+\Big(1+2C_m^2\big|\pi_{k+1}(\Omega_1)-\pi_{k}(\Omega_1)\big|\Big)^{\frac12}\\
       \label{eq:thetaj2} &\leq ~\quad\frac{C_r^{\frac12}}{M}+1+C_m^2\big|\pi_{k+1}(\Omega_1)-\pi_{k}(\Omega_1)\big|.
\end{align}
For term~$J_4$,
\begin{equation}\label{eq:thetaj4}
    J_4\leq \abs[\Big]{ \int_{K}\psi_{2,k+1}\pi_{k+1}\, d x}\cdot \abs[\Big]{ \int \psi_{2,k}\pi_{k+1}\, dx}\leq J_1.
\end{equation}
Hence,
\begin{align}
       \MoveEqLeft\Theta(k,k+1)\overset{\mathclap{\eqref{eq:thetarewrite},\eqref{eq:ckT1}}}{\leq} ~J_1(C_{\psi}+1)+J_2\overset{\eqref{eq:thetaj4}}{\leq} J_1(C_{\psi}+2)+J_3\\
       \label{eq:thetakk1} &\overset{\mathclap{\eqref{eq:thetaj1},\eqref{eq:thetaj2}}}{\leq}  \quad 1+(2C_m(C_{\psi}+2)+C_m^2)\cdot \big|\pi_{k+1}(\Omega_1)-\pi_{k}(\Omega_1)\big|+\frac{(C_{\psi}+2)C_r+C_r^{\frac12}}{M}.
\end{align}

\step[Estimating~$\prod_{j=k}^{M-1}\Theta(j,j+1)$]
By direct computation, for~$k\geq k_0$,
\begin{align}
 & \prod_{j=k}^{M-1}\beta_j\leq
  \prod_{j=k}^{M-1}\Theta(j,j+1)\\
&\overset{\mathclap{\eqref{eq:thetakk1}}}{\leq} ~\prod_{j=k}^{M-1}\paren[\Big]{1+(2C_m(C_{\psi}+2)+C_m^2)\cdot \big|\pi_{k+1}(\Omega_1)-\pi_{k}(\Omega_1)\big|
+\frac{(C_{\psi}+2)C_r+C_r^{\frac12}}{M}}\\
       &\overset{\mathclap{\text{AM-GM}}}{\leq}  ~~ \bigg(
	1+\frac{(C_{\psi}+2)C_r+C_r^{\frac12}}{M}
    \\
      &\qquad\qquad
    \mathbin{+} \frac{1}{M-k}\sum\limits_{j=k}^{M-1}\paren[\big]{2C_m(C_{\psi}+2)+C_m^2}\cdot \big|\pi_{j+1}(\Omega_1)-\pi_{j}(\Omega_1)\big|
       \bigg)^{M-k}\\
        &\overset{\mathclap{\eqref{eq:massBV}}}{\leq}\quad \paren*{  1+\frac{C_{\BV}(2C_m(C_{\psi}+2)+C_m^2) }{M-k}+\frac{(C_{\psi}+2)C_r+C_r^{\frac12}}{M}}^{M-k}
        \overset{\eqref{eq: defCbeta}}{\leq} C_{\beta},\label{eq:jgeqk0beta}
\end{align}
where the last inequality uses the fact that~$M-k\leq M$. 
\end{proof}

\subsection{Proof of Lemma~\ref{l:iteration}}
\label{s:proofIteration}

In this section, we show that by choosing proper~$N$ and~$T$, the probability~$\theta_k$ defined in \eqref{e:probIteration} can be lower bounded by~$1-\frac{\delta}{M}$.  The proof is a direct calculation based on several lemmas.

We first choose sample size~$N$ such that when \eqref{eq:achoose} holds, the relevant terms in \eqref{e:probIteration} is small enough. Then for a fixed sample size, we then choose the temperature such that the first term in \eqref{e:probIteration} is small enough. Those choices are stated in the following two lemmas, whose proof is direct calculation and thus omitted.
\begin{lemma}[Choose sample size~$N$] \label{lem:Nchoice}
Given~$\theta\in (0,1)$ and~$\delta>0$. Let ~$a_1, a_2, a_3$ be selected as \eqref{eq:achoose}. if choose~$N$ as \eqref{eq:Nchoice},  then we have
\begin{equation}\label{eq:probNa}
    \exp\paren[\Big]{-\frac{2Na_1^2}{C_r^2}}\leq \frac{\theta}{8M},\quad 
    \exp\paren[\Big]{-\frac{Na_2^2}{2C_{r}C_{\psi}^2}}\leq \frac{\theta}{8M},\quad
    \exp\paren[\Big]{-\frac{Na_3^2}{2C_{\psi}^2}}\leq  \frac{\theta}{8M}.
\end{equation}
\end{lemma}

\begin{lemma}\label{lem:etaCriProb}
    Fix~$\theta\in (0,1)$ and~$\delta>0$ small, choose sample size~$N$ according to \eqref{eq:Nchoice}. Then for every~$\eta_k$ such that
    \begin{equation}\label{eq:etaCriN}
        \eta_k\leq \frac{C_K}{\log\paren[\Big]{\frac{8MNC_P(C_{\psi}+1)}{\theta}}}\, ,
    \end{equation}
    we have
    \begin{equation}\label{e:probsmall}
        \exp\paren[\Big]{-\frac{C_{K}}{\eta_k}}\leq \frac{\theta}{8MNC_P(C_{\psi}+1)}.
    \end{equation}
\end{lemma}

Finally, we choose the proper running time depending on the sample size such that the second term in \eqref{e:probIteration} is small enough.
\begin{lemma}\label{lem:Tlocal}
    For
    \begin{equation}\label{eq:TCriProb}
         T\geq \max\Big\{\frac{2}{\Lambda}\Big( \frac{\hat{U}}{2\eta}+\log \paren[\Big]{\frac{8M}{\theta}}+\frac{1}{8}+\log(N)\Big), 
\frac{c}{d}\log (4d)\Big\},
    \end{equation}
    we have
    \begin{equation}\label{eq:TCriprob1}
       e^{-\Lambda T/2}\max_{x\in K}\sqrt{\frac{p_{\epsilon,T}(x,x)}{\pi(x)}-1}\leq \frac{\theta}{8MN}.
    \end{equation}
\end{lemma}
\begin{proof}[Proof of Lemma~\ref{lem:Tlocal}]
When~$T$ satisfies \eqref{eq:TCrip}, we have \eqref{eq:ubPatTcri}.
    Thus when~$T$ satisfies \eqref{eq:TCriProb}, we have
\begin{equation}
    e^{-\Lambda T/2}\leq \exp\paren[\Big]{\frac{\hat{U}}{2\eta_k}}  e^{-\frac{1}{8}}\frac{\theta}{8MN}
\end{equation}
which implies that
\begin{equation}
    e^{-\Lambda T/2}\max_{x\in K}\sqrt{\frac{p_{\epsilon,T}(x,x)}{\pi(x)}-1}\overset{\eqref{eq:poverpiK},\eqref{eq:ubPatTcri}}{\leq} \frac{\theta}{4MC_r^{\frac12}(C_{\psi}+1)}. \qedhere
\end{equation}
\end{proof}
Combining the above lemmas, we now state and prove the estimates for~$\theta_k$, completing the proof of Lemma~\ref{l:iteration}.

\begin{proof}[Proof of Lemma~\ref{l:iteration}] Fix~$\alpha,\theta,\delta>0$. First observe that according to Lemma \eqref{lem: iterative scheme between levels}, the event \eqref{eq:Xk+1inK} the recurrence relation \eqref{e:iteration} holds with probability~$\theta_k$. It remains to show that for the proper~$\eta_{\mathrm{cr}}$,~$N$ and~$T$, we have that \eqref{eq:thetakLB} and \eqref{e:CBeta} hold.
    \begin{equation}\label{eq:etacrchoice}
     \eta_{\mathrm{cr}}\defeq \min\Bigl\{ \frac{C_K}{\log\paren[\Big]{\frac{8MNC_P(C_{\psi}+1)}{\theta}}},  \frac{C_K}{\log\paren{{C}_{\gamma}M}}\Bigr\}.
\end{equation}
Here~$C_K$ is the constant defined in \eqref{eq:defCK} with the fixed~$\alpha>0$ and~$C_\gamma$ is the constant defined in \eqref{e:dEF} when~$\gamma\defeq C_K$.
Notice that if~$\eta_{\mathrm{cr}}$ are chosen according to~\eqref{eq:etacrchoice}, then we can find dimensional constants~$C_{\mathrm{tem}} = C_{\mathrm{tem}}(\alpha,U)>0$ so that this choice is consistent with the choice in~\eqref{e:CriEtaLT}.

Choose~$N$ as in \eqref{eq:Nchoice} and~$T$ as
\begin{align}\label{e:Tchoicelem}
     T&\geq \max\Big\{\frac{2}{\Lambda}\Big(\log\paren[\Big]{\frac{1}{\delta}}+ \frac{\hat{U}}{2\eta}+\log M+\frac{1}{8}+\log(4C_p^{\frac12}C^{\frac12}_{r}(C_{\psi}+1))\Big), \\
     &\qquad\qquad\frac{2}{\Lambda}\Big( \frac{\hat{U}}{2\eta}+\log \paren[\Big]{\frac{8M}{\theta}}+\frac{1}{8}+\log(N)\Big),\frac{c}{d}\log (4d)\Big\}
\end{align}
Notice that if~$T$ is chosen according to~\eqref{e:Tchoicelem}, then we can find constants~$C_{\alpha} = C_{\alpha}(\alpha,U)>0$ so that this choice is consistent with the choice in~\eqref{e:TNlow}.
 
When~$\eta_k\leq \eta_{\mathrm{cr}}$,~$k$ also larger than the~$k_0$ defined in \eqref{eq:defEtacr}. Thus Lemma~\ref{lem: betaj} gives that the first inequality in \eqref{e:CBeta} holds.

 On the other hand, when~$\eta_k\leq \eta_{\mathrm{cr}}$, we have that~$\eta_k$ satisfies \eqref{e:etaCriC}. Meanwhile, the choice of~$T$ satisfies \eqref{eq:Tck}. Thus with the choice \eqref{eq:achoose}, Lemma~\ref{lem: cj} gives that the second inequality in \eqref{e:CBeta} holds.

With the choice of~$a_1,a_2,a_3$ as in \eqref{eq:achoose} and~$N$ as in \eqref{eq:Nchoice}, Lemma~\ref{lem:Nchoice} gives that
\eqref{eq:probNa} holds. According to Lemma~\ref{lem:etaCriProb}, since~$\eta_k$ also satisfies \eqref{eq:etaCriN}, the inequality \eqref{e:probsmall} holds. Finally, the choice of~$T$ in \eqref{e:Tchoicelem} also satisfies \eqref{eq:TCriProb}, thus \eqref{eq:TCriprob1} holds as well. Plugging the above estimates into the definition of~$\theta_k$ yields 
  \begin{align}
      \theta_k&\overset{\mathclap{\eqref{eq:probNa}}}{\geq}1-N\paren[\bigg]{C_P(C_{\psi}+1)\exp\paren[\Big]{-\frac{C_{K}}{\eta_k}}-e^{-\Lambda T/2}\max_{x\in K}\sqrt{\frac{p_{k,T}(x,x)}{\pi_{k}(x)}-1}}-\frac{3\theta}{4}\\
      &\overset{\mathclap{\eqref{e:probsmall},\eqref{eq:TCriprob1}}}{\geq}
       \quad 1 -\frac{\theta}{4M}-\frac{3\theta}{4M}=1-\frac{\theta}{M}. \qedhere
  \end{align}
\end{proof}

\section{Error estimates in the high temperature regime (Lemma~\ref{lem:globalmixing} and~\ref{lem:baseCasePsi})} \label{sec:global}
In this section we prove Lemma~\ref{lem:globalmixing} and~\ref{lem:baseCasePsi}. The proof is inspired by \cite{Marion23SMC}, where they assume at the first level~$X_{1,T}^{i}$ are i.i.d exactly following~$\mu_1$. Here we briefly introduce the proof strategy. 
At each level~$k$, we use a maximal coupling approach to construct a set of random variables~$\bar{X}_{k,T}^{i}$ having marginal distribution \emph{exactly} the distribution~$\pi_k$.  
The construction of~$\bar{X}_{k,T}^{i}$ is given explicitly in \cite[Appendix]{Marion23SMC} and it is shown there that for every~$i=1,\dots, N$,
\begin{equation}\label{eq:propTV}
    P({X}_{k,T}^{i}\neq\bar{X}_{k,T}^{i})=\norm{p_{k,T}(X_{1,0}^{i},\cdot)-\pi_k(\cdot)}_{\mathrm{TV}}.
\end{equation}
Here we view~$X_{1,0}^{i}$ as a fixed point (starting point of the Langevin dynamics on level~$k$). 

Heuristically, if at level~$k$ the right hand side of \eqref{eq:propTV} is small, then with high probability, the particles can be coupled to a set of particles drawn exactly from the target distribution~$\pi_k$. An inductive step can then be established that these conditions hold as long as we have good local mixing. We will show that this can be done in the high temperature regime, i.e.~$k\leq k_{\mathrm{cr}}$, where the global mixing time is not too long. 

We now illustrate the above heuristic in more details and begin by introducing the notations, which is consistent with those in \cite{Marion23SMC}. 
For steps~$k=1,\dots, k_{\mathrm{cr}}$, define the events
\begin{align}
    \mathbf{A}_k&=\set{X_{k,T}^{i}=\bar{X}_{k,T}^{i},\forall i}.\\
    \mathbf{B}_k&=\set[\Big]{\abs[\Big]{\frac{1}{N}\sum_{i=1}^{N}\tilde{r}_k(X_{k,T}^{i})-\int \tilde{r_k}\pi_k\, d x}\leq \paren[\Big]{\int \tilde{r_k}\pi_k\, d x}/3}.\\
    \mathbf{C}_k&=\mathbf{A}_k\cap \mathbf{B}_k.
\end{align}
Each~$\mathbf{A}_k$ is the coupling event between the particles~$X_{k,T}^i$ and the particles~$\bar{X}_{k,T}^{i}$ drawn independently from~$\pi_k$. The event~$\mathbf{B}_k$ represents that the error of the empirical estimator at level~$k$ is within a reasonable interval. In \cite{Marion23SMC}, the following inductive relations are established. For completeness, we provide a short proof.
\begin{lemma}[Lemma 4.4 in \cite{Marion23SMC}]\label{lem:MMS44}
     Suppose~$\P(C_{k-1})\geq \frac{3}{2\omega}$ for some~$\omega\geq \frac{3}{2}$. For any measurable~$g$ such that~$|g|\leq G$. Fix~$\theta\in (0,1)$ and~$\theta'\in (0,1)$ and~$\mathfrak{e}>0$, then for any 
    \begin{equation}\label{e:NTgloballem}
        N\geq \frac{G^2}{2\mathfrak{e}^2}\log\paren[\Big]{\frac{2}{\theta'}},\quad \text{and}\quad T\geq\frac{1}{\lambda_{2,k}}\paren[\Big]{\log\paren[\Big]{\frac{2N}{\theta}}+\log(\omega-1)},
    \end{equation}
    we have
    \begin{equation}
        \P\paren[\Big]{\abs[\Big]{\frac{1}{N}\sum_{i=1}^{N}g(X_{k,T}^{i})-\int g\pi_{k}\, d x}<\mathfrak{e}}\geq (1-\theta)\cdot\P(C_{k-1})-\theta'.
    \end{equation}
\end{lemma}

\begin{lemma}[Corollary 5.1 in \cite{Marion23SMC}] \label{lem:MMScor41}
    Suppose~$\P(\mathbf{C}_k)\geq \frac{3}{2\omega}$ for some~$\omega\geq \frac{3}{2}$. Fix~$\theta\in (0,1)$ and~$\theta'\in (0,1)$, then for any 
    \begin{equation}\label{eq:NTglobalcor}
        N\geq \frac{9C_r^2}{2}\log\paren[\Big]{\frac{2}{\theta'}},\quad \text{and}\quad T\geq\frac{1}{\lambda_{2,k}}\paren[\Big]{\log\paren[\Big]{\frac{2N}{\theta}}+\log(\omega-1)},
    \end{equation}
    we have
    \begin{equation}
        \P(\mathbf{C}_{k})\geq (1-\theta)\cdot\P(\mathbf{C}_{k-1})-\theta'.
    \end{equation}
\end{lemma}
\begin{proof}[Proof of Lemma~\ref{lem:MMS44} and~\ref{lem:MMScor41}]
   The statement in Lemma~\ref{lem:MMS44} and~\ref{lem:MMScor41} is identical to \cite[Lemma 4.4 and Corollary 5.1]{Marion23SMC} except in \eqref{e:NTgloballem} and \eqref{eq:NTglobalcor}, they use the expression
   \begin{equation}
       T\geq \tau_{k}\paren[\Big]{\frac{\theta}{N},\omega},
   \end{equation}
   where~$\tau_k$ is a weaker notion of the commonly-used mixing time, see \cite[Section 2.1]{Marion23SMC} for rigorous definition. 
  It is shown in \cite[Section 2.2]{Marion23SMC} that~$\tau_k$ satisfies the following bound 
   \begin{equation}
       \tau_{k}(\iota,\omega)\leq \frac{1}{\rho_k}\paren[\Big]{\log\paren[\Big]{\frac{2}{\iota}}+\log(\omega-1)}.
   \end{equation}
Here~$\rho_k$ denotes the spectral gap of the transition kernel of Langevin dynamics at level~$k$, which is~$\lambda_{2,k}$. Taking~$\iota=\theta/N$ finishes the proof.
\end{proof}

Lemma~\ref{lem:MMS44} and~\ref{lem:MMScor41} provide the estimates in the inductive steps, we still need the estimate at the first level. The following lemma shows that we can still satisfy the base case provided proper sample size and running time.
\begin{lemma}[Mixing at the first level] \label{lem:firstlevel}
    Fix~$\theta'\in (0,1)$ and let assumption~\ref{assum:startgood} hold. If
    \begin{align}\label{eq:NT1st}
        N&\geq 18C_r^2\log\paren[\Big]{\frac{4}{\theta'}}\\
         T&\geq \max\Big\{\frac{2}{\lambda_{2,1}}\Big( \frac{C_{\mathrm{ini}}}{2}
         +\frac{1}{8}+\frac12\log(C_p) + \log\paren[\Big]{\frac{2N}{\theta'}}\Big),
\frac{c}{d}\log (4d)\Big\}
    \end{align}
    we have
    \begin{equation}
        \P(\mathbf{C}_1)\geq 1-\theta'.
    \end{equation}
\end{lemma}

Before proving Lemma~\ref{lem:firstlevel}, we state a fact about 
the convergence of~$p_{\epsilon,t}$ toward the global modes (the stationary distribution~$\pi_{\epsilon}$), as a complement to Lemma~\ref{eq:TDquick}.
\begin{lemma}\label{lem:TDslow}
    Let~$p_{\epsilon,t}(x,y)$ be the transition kernel defined as in \eqref{eq:defp}, then
    \begin{equation}\label{eq:TDslow}
    \norm[\Big]{\frac{p_{\epsilon,t}(x,\cdot)}{\pi_{\epsilon}(\cdot)}-1}_{L^2(\pi_{\epsilon})}\leq e^{-\lambda_{2,\epsilon} t/2}\sqrt{\frac{p_{\epsilon,t}(x,x)}{\pi_{\epsilon}(x)}-1}\,.
\end{equation}
\end{lemma}

\begin{proof}[Proof of Lemma~\ref{lem:TDslow}]
As in the proof of Lemma~\ref{lem:TDquick}, since~$\epsilon$ is fixed, for simplicity of presentation, we slightly abuse notation in this proof and omit the subscript~$\epsilon$. The proof follows similarly as that of Lemma~\ref{lem:TDquick}. Except for \eqref{eq:PoverPi}, we rewrite as
\begin{align}
    \frac{p_t(x,y)}{\pi(y)}-1 \,&\overset{\mathclap{\eqref{eq:trans1}}}{=}\int \paren[\Big]{\frac{p_s(x,z)}{\pi(z)}-1}\frac{p_{t-s}(y,z)}{\pi(z)}\pi(z)\, d z\\
    &\overset{\mathclap{\eqref{eq:decomP}}}{=}\int \paren[\Big]{\frac{p_s(x,z)}{\pi(z)}-1}\pi(z)\, d z\\
    &+\sum_{k\geq 2}e^{-\lambda_k (t-s)}\psi_{k}(y) \int\paren[\Big]{\frac{p_s(x,z)}{\pi(z)}-1}\psi_k(z)\pi(z)\, d z.  
\end{align}
It follows immediately that
\begin{equation}\label{eq:decom3}
    \norm[\Big]{\frac{p_t(x,\cdot)}{\pi(\cdot)}-1}_{L^2(\pi)}\leq e^{-\lambda_2 (t-s)}\norm[\Big]{\frac{p_s(x,\cdot)}{\pi(\cdot)}-1}_{L^2(\pi)}.
\end{equation}
Choosing~$s=t/2$ and plugging \eqref{eq:psL2} into \eqref{eq:decom3} finishes the proof.
\end{proof}

We now prove Lemma~\ref{lem:firstlevel}.

\begin{proof}[Proof of Lemma~\ref{lem:firstlevel}]

\restartsteps

\step[Bound~$\P(\mathbf{A}_1)$]
Following the explicit maximal coupling construction of~$\bar{X}^{i}_{1,T}$ given in \cite[Appendix]{Marion23SMC}, we obtain that for every~$i=1,\dots,N$,
\begin{equation}
    \P(X_{1,T}^{i}\neq \bar{X}^{i}_{1,T})=\norm{p_{1,T}(X_{1,0}^{i},\cdot)-\pi_1(\cdot)}_{\mathrm{TV}}.
\end{equation}
Observe that for any~$x\in\mathbb{R}^d$,
\begin{align}
  \norm{p_{1,T}(x,\cdot)-\pi_1(\cdot)}_{\mathrm{TV}}&\leq \norm[\Big]{\frac{p_{1,T}(x,\cdot)}{\pi_1(\cdot)}-1}_{L^2(\pi_1)}\\
  &\overset{\mathclap{\eqref{eq:TDslow}}}{\leq} \quad e^{-\lambda_{2,1} T/2}\sqrt{\frac{p_{1,t}(x,x)}{\pi(x)}-1}\\
  &\overset{\mathclap{\eqref{eq:poverpi}}}{\leq}\quad e^{-\lambda_{2,1} T/2}C^{\frac12}_{p}\exp\paren[\Big]{\frac{U(x)}{2}}\paren[\Big]{1-\exp(-\frac{c}{d}t)}^{-\frac{d}{4}}\\
  &\overset{\mathclap{\eqref{eq:TCrip},\eqref{eq:ubPatTcri}}}{\leq}\quad e^{-\lambda_{2,1} T/2}C^{\frac12}_{p}\exp\paren[\Big]{\frac{U(x)}{2}}e^{\frac{1}{8}}.
\end{align}
Therefore, according to Assumption~\ref{assum:startgood}, we have that
\begin{equation}\label{eq:tv1st}
    \norm{p_{1,T}(X_{1,0}^{i},\cdot)-\pi_1(\cdot)}_{\mathrm{TV}} \leq e^{-\lambda_{2,1} T/2}C^{\frac12}_{p}\exp\paren[\Big]{\frac{C_{\mathrm{ini}}}{2}}e^{\frac{1}{8}}
    \leq \frac{\theta'}{2N}.
\end{equation}
Then, by union bound,
\begin{equation}\label{eq:PA1}
    \P(\mathbf{A}_1)\geq 1-\sum_{i=1}^{N} \P(X_{1,T}^{i}\neq \bar{X}^{i}_{1,T}) \overset{\eqref{eq:tv1st}}{\geq} 1-\frac{\theta'}{2}.
\end{equation}
\step[Bound~$\P(\mathbf{B}_1)$] Observe that by scaling, the event~$\mathbf{B}_1$ is equivalent to 
\begin{align}
    \MoveEqLeft\set[\Big]{\abs[\Big]{\frac{1}{N}\sum_{i=1}^{N}r_1(X_{1,T}^{i})-\int r_1\pi_1\, d x}\leq \paren[\Big]{\int r_1\pi_1\, d x}/3}\\
    &=\set[\Big]{\abs[\Big]{\frac{1}{N}\sum_{i=1}^{N}r_1(X_{1,T}^{i})-1}\leq 1/3}. \label{eq:B1equiv}
\end{align}
Notice that
\begin{align}
    \abs[\Big]{\frac{1}{N}\sum_{i=1}^{N}r_1(X_{1,T}^{i})-1}&\leq  \abs[\Big]{\frac{1}{N}\sum_{i=1}^{N}r_1(X_{1,T}^{i})-\E\frac{1}{N}\sum_{i=1}^{N}r_1(X_{1,T}^{i})}\\
    &+\abs[\Big]{\E\frac{1}{N}\sum_{i=1}^{N}r_1(X_{1,T}^{i})-1}.
\end{align}
The first term,  by Hoeffding's equality, since~$0<r_1\leq C_r$, 
\begin{align}
    \P\paren[\Big]{\abs[\Big]{\frac{1}{N}\sum_{i=1}^{N}r_1(X_{1,T}^{i})-\E\frac{1}{N}\sum_{i=1}^{N}r_1(X_{1,T}^{i})}>\frac{1}{6}}&\leq 2\exp\paren[\Big]{-\frac{N}{18 C_r^2}}\overset{\mathclap{\eqref{eq:NT1st}}}{\leq}~\frac{\theta'}{2}. \label{eq:B1term1}
\end{align}
For the second term, we observe that
\begin{align}
   \abs[\Big]{\E\frac{1}{N}\sum_{i=1}^{N}r_1(X_{1,T}^{i})-1}&\leq \frac1N\sum_{i=1}^{N}\abs[\Big]{\E r_1(X_{1,T}^{i})-1}\\
    &\leq \frac1N\sum_{i=1}^{N}\norm{r_1}_{L^{\infty}}\norm{p_{1,T}(X_{1,0}^{i},\cdot)-\pi_1(\cdot)}_{\mathrm{TV}}\\
    &\overset{\mathclap{\eqref{eq:tv1st}}}{\leq} ~\frac{C_r \theta'}{2N}\leq \frac{C_r}{2N} \overset{\eqref{eq:NT1st}}{\leq}\frac{1}{6}. \label{eq:B1term2}
\end{align}

Therefore, 
\begin{equation}
    \P(\mathbf{B}_1)=\P\paren[\Big]{\abs[\Big]{\E\frac{1}{N}\sum_{i=1}^{N}r_1(X_{1,T}^{i})-1}<\frac{1}{3}}\overset{\eqref{eq:B1term1},\eqref{eq:B1term2}}{\geq} 1-\frac{\theta'}{2}. \label{eq:PB1}
\end{equation}
Combining the two steps above, we obtain that
\begin{equation}
    \P(\mathbf{C}_1)\overset{\eqref{eq:PA1},\eqref{eq:PB1}}{\geq} 1-\theta'.\qedhere
\end{equation}
\end{proof}

\subsection{Estimate for the global mixing (Lemma~\ref{lem:globalmixing})}
In this section we prove Lemma~\ref{lem:globalmixing} following the induction strategy of \cite[Theorem 3.1]{Marion23SMC}. The difference is that in \cite[Theorem 3.1]{Marion23SMC}, they assume~$X_{1,T}^i$ are i.i.d following~$\pi_1$ whereas here we use Lemma~\ref{lem:firstlevel} to guarantee the base case of the induction.

Before the proof of Lemma~\ref{lem:globalmixing}, we need to following lower bound of second eigenvalue, which can be obtained immediately by  Corollary 2.15 in \cite{MenzSchlichting14}. The proof is identical to \cite[Lemma 7.8]{han2025polynomialcomplexitysamplingmultimodal} and thus we omit it here.
\begin{lemma}[Lemma 7.8 in \cite{han2025polynomialcomplexitysamplingmultimodal}, Corollary 2.15 in \cite{MenzSchlichting14}]\label{l: lower_bound_next_eigenvalue}
  Suppose~$U$ satisfies assumptions~\ref{a:criticalpts} and~\ref{assumption: nondegeneracy}, and recall~$\hat U$ is the saddle height defined in~\eqref{e:UHatDef}.
    For every~$H>\hat{U}$, there exists~$A\defeq A(H,d,U)>0$ independent of~$\epsilon$ such that for every~$\epsilon<1$,
\begin{equation}\label{eq:lambdaLowBound}
  \lambda_{2,\epsilon}\geq A\exp\paren[\Big]{-\frac{H}{\epsilon}}.
\end{equation}
\end{lemma}

Now we are well-equipped to prove Lemma~\ref{lem:globalmixing}.
\begin{proof}[Proof of Lemma~\ref{lem:globalmixing}]
Fix~$2\leq k\leq k_{\mathrm{cr}}$. Observe that when~$\eta_k\geq \eta_{\mathrm{cr}}$,
\begin{equation}
    \lambda_{2,k}\overset{\mathclap{\eqref{eq:lambdaLowBound}}}{\geq}A_{\alpha}\exp\paren[\Big]{-\frac{(1+\alpha)^{\frac12}\hat{U}}{\eta_{\mathrm{k}}}}\geq  A_{\alpha}\exp\paren[\Big]{-\frac{(1+\alpha)^{\frac12}\hat{U}}{\eta_{\mathrm{cr}}}}.
\end{equation}
where~$A_{\alpha}$ is the constant in \eqref{eq:lambdaLowBound} with~$H=(1+\alpha)^{\frac12}\hat{U}$.
Define
\begin{align}
   T&\geq \max\Bigl\{\frac{1}{A_{\alpha}}\exp\paren[\Big]{\frac{(1+\alpha)^{\frac12}\hat{U}}{\eta_{\mathrm{cr}}}}\paren[\Big]{\log(2N)+\log\paren[\Big]{\frac{2M}{p}} },\\
   &\qquad\frac{2}{\lambda_{2,1}}\Big( \frac{C_{\mathrm{ini}}}{2}
         +\frac{1}{8}+\frac12\log(C_p) + \paren[\Big]{\log(2N)+\log\paren[\Big]{\frac{4M}{p^2}}}\Big),
\frac{c}{d}\log (4d)\Bigl\}. 
    \end{align}
Then we can find constant~$\tilde{C}_{\alpha} = \tilde{C}_{\alpha}(\alpha, U, C_{\mathrm{ini}})>0$ so that this choice is consistent with the choice in~\eqref{eq:NTglobal}.

Notice that~$T$ satisfies \eqref{e:NTgloballem}, \eqref{eq:NTglobalcor} \eqref{eq:NT1st} with
\begin{equation}\label{eq:ppomegachoice}
    \theta=\frac{p}{2M},\quad \theta'=\theta^2=\frac{p^2}{4M},\quad\omega=2.
\end{equation}
Observe that~$N$ also satisfies \eqref{e:NTgloballem}, \eqref{eq:NTglobalcor} \eqref{eq:NT1st} with the same choice of~$\theta, \theta'$ and~$\omega$.

By Lemma~\ref{lem:MMS44}, 
for function~$f$ with~$\abs{f}\leq 1$, we have that since~$N$ satisfies \eqref{e:NTgloballem},
\begin{equation}
     \P\paren[\Big]{\abs[\Big]{\frac{1}{N}\sum_{i=1}^{N}f(X_{k,T}^{i})-\int f\pi_{k}\, d x}<\mathfrak{e}} \geq (1-\theta)\P(\mathbf{C}_{k-1})-\theta'.
\end{equation}
The probability~$\P(\mathbf{C}_{k-1})$ can be lower bounded by induction using Lemma~\ref{lem:MMScor41}. Notice that the choice of~$N,T$ in \eqref{eq:NTglobal} ensures that~$\P(\mathbf{C}_1)\geq 1-\theta'$. Then applying Lemma~\ref{lem:MMScor41} repeatedly gives that 
\begin{align}
    \P(C_k)&\geq (1-\theta)^{k-1}\cdot (1-\theta')-\theta'\sum_{s=0}^{k-2}(1-\theta)^{s}\\
    &=(1-\theta)^{k-1}-\theta'\sum_{s=0}^{k-1}(1-\theta)^{s}. \label{eq:PCk}
\end{align}
On one hand, this gives that
\begin{equation}
   \P(C_k) \geq (1-\theta)^{k-1}-\theta'\frac{1-(1-\theta)^{k}}{\theta} \geq (1-\theta)^{M}-\frac{\theta'}{\theta}\overset{\eqref{eq:ppomegachoice}}{\geq} 1-p\geq \frac{3}{4}
\end{equation}
showing that the conditions in Lemma~\ref{lem:MMS44} and~\ref{lem:MMScor41} satisfies with~$\omega=2$.
On the other hand, the inequality \eqref{eq:PCk} gives that for~$2\leq k\leq k_{\mathrm{cr}}$, ~$\abs{f}\leq 1$,
\begin{align}
    \P\paren[\Big]{\abs[\Big]{\frac{1}{N}\sum_{i=1}^{N}f(X_{k,T}^{i})-\int f\pi_{k}\, d x}<\mathfrak{e}} &\geq (1-\theta)\P(\mathbf{C}_{k-1})-\theta'\\
    &\geq (1-\theta)^{k}-\theta'\sum_{s=0}^{k}(1-\theta)^{s}\\
    &\geq (1-\theta)^{M}-\frac{\theta'}{\theta}\overset{\eqref{eq:ppomegachoice}}{\geq} 1-p. \qedhere
\end{align}
\end{proof}

\subsection{Estimates of the critical state (Lemma~\ref{lem:baseCasePsi})}
In this section we prove Lemma~\ref{lem:baseCasePsi}, which is straightforward using Lemma~\ref{l:rebalancing} and Lemma~\ref{lem:globalmixing}.
\begin{proof}[Proof of Lemma~\ref{lem:baseCasePsi}]
We apply Lemma~\ref{l:rebalancing} with 
\begin{equation}
    a=\frac{\delta}{8M},\quad h=\psi_{2,k_{\mathrm{cr}}}\one_{K}\quad r = r_{k_{\mathrm{cr}}-1}, \quad x^i=X^i_{k_{\mathrm{cr}}-1,T},\quad y^i=X^i_{k_{\mathrm{cr}},0}
\end{equation}
and obtain that with probability larger than or equal to
\begin{equation}\label{e:lem37prob1}
    1-2\exp\paren[\Big]{-\frac{2N a^2}{\norm{\psi_{2,k_{\mathrm{cr}}}\one_{K}}^2_{\osc}}}\geq  1-2\exp\paren[\Big]{-\frac{N a^2}{2C^2_{\psi}}}\overset{\eqref{eq:Nchoice}}{\geq} 1-\frac{\theta}{4M}.
\end{equation}
we have that
\begin{equation}\label{eq:psicr}
   \abs[\Big]{\frac{1}{N}\sum_{i=1}^{N}\psi_{2,k_{\mathrm{cr}}}\one_{K}(X^{i}_{k_{\mathrm{cr}},0})}\leq  (C_{\psi}+1)J_1 +J_2
     +\frac{\delta}{8M}+J_3. 
  \end{equation}
where
\begin{align}
    J_1&=\abs[\Big]{1-\frac{1}{N}\sum_{i=1}^N r_{k_{\mathrm{cr}}-1}(X_{k_{\mathrm{cr}}-1,T}^i)},\\
    J_2&=\abs[\Big]{\frac{1}{N}\sum_{i=1}^N r_{k_{\mathrm{cr}}-1}(X_{k_{\mathrm{cr}}-1,T}^i)\Big(\psi_{2,k_{\mathrm{cr}}}\one_{K}(X^{i}_{k_{\mathrm{cr}},0})  -\int_{\mathcal X} \psi_{2,k_{\mathrm{cr}}}\one_{K} \pi_{k_{\mathrm{cr}}} \, dx \Big)},\\
    J_3&=\abs[\Big]{\int_{\mathcal X} \psi_{2,k_{\mathrm{cr}}}\one_{K} \pi_{k_{\mathrm{cr}}} \, dx }.
\end{align}
Here we use the fact that \begin{equation}\label{eq:J1Linfity}
  \norm[\Big]{\psi_{2,k_{\mathrm{cr}}}\one_{K}-\int_{K}\psi_{2,k_{\mathrm{cr}}}\pi_{k_{\mathrm{cr}}}\, d x}_{L^\infty}
   \overset{\eqref{eq:ckT1}}{\leq} C_{\psi}+1. 
\end{equation}

Next, we bound~$J_1$,~$J_2$ and~$J_3$, respectively. Notice that our choice of~$N$ satisfies \eqref{eq:NTglobal} with
\begin{equation}\label{e:pfrake}
\mathfrak{e}=\frac{\delta}{2(C_{\psi}+1)}    \quad p=\theta.
\end{equation}
Indeed,
\begin{align*}
    N\geq 64C^2_r(C_{\psi}+1)^2 \frac{M^2}{\delta^2}\log\paren[\Big]{\frac{(8M)^2}{\theta^2}} \geq \frac{16 C_r^2}{\mathfrak{e}^2}\log\paren[\Big]{\frac{64M^2}{\theta^2}}.
\end{align*}
Similarly,~$T$ satisfies \eqref{eq:NTglobal} with the same choice \eqref{e:pfrake} and~$\eta_{\mathrm{cr}}$ as in \eqref{e:CriEtaLT}.
Therefore, according to Lemma~\ref{lem:globalmixing},
\begin{equation}
     \P(J_1<C_r\mathfrak{e})\geq 1-\theta,\quad
    \P(J_2< C_r(C_{\psi}+1)\mathfrak{e})\geq 1-\theta,
\end{equation}
where we use the fact that~$\norm{r_{k_{\mathrm{cr}}-1}}_{L^{\infty}}\leq C_r$ and \eqref{eq:J1Linfity}. That is,
\begin{equation}\label{e:J1J2prob}
     \P\paren[\Big]{J_1<\frac{C_r\delta}{2(C_{\psi}+1)}}\geq 1-\theta,\quad
    \P\paren[\Big]{J_2< \frac{C_r\delta}{2}}\geq 1-\theta.
\end{equation}

For term~$J_3$, we first notice that~$\eta_{\mathrm{cr}}$ is actually chosen according to \eqref{eq:etacrchoice}. This~$\eta_{\mathrm{cr}}$ also satisfies \eqref{eq:etaCriN}. We apply Lemma~\ref{lem:etaCriProb} and obtain that
\begin{equation}\label{e:J3prob}
    J_3\overset{\eqref{e:intKpsipi}}{\leq} C_P \exp\paren{-\frac{C_K}{\eta_{\mathrm{cr}}}}\overset{\eqref{e:probsmall}}{\leq}\frac{\theta}{8MN(C_{\psi}+1)}\overset{\eqref{eq:Nchoice}}{\leq} \frac{\delta}{8M(C_{\psi}+1)}.
\end{equation}
Combining the above discussion, we obtain from \eqref{e:lem37prob1} and  \eqref{e:J1J2prob} that with probability larger than or equal to 
\begin{equation}
    \paren[\Big]{1-\frac{\theta}{4M}}\cdot \paren[\Big]{1-2\theta}\geq 1-3\theta,
\end{equation}
we have
\begin{align}
    \abs[\Big]{\frac{1}{N}\sum_{i=1}^{N}\psi_{2,k_{\mathrm{cr}}}\one_{K}(X^{i}_{k_{\mathrm{cr}},0})}~&\overset{\mathclap{\eqref{eq:psicr}}}{\leq}  \quad (C_{\psi}+1)J_1 +J_2
     +\frac{\delta}{8M}+J_3\\
     &\overset{\mathclap{\eqref{e:pfrake},\eqref{e:J1J2prob},\eqref{e:J3prob}}}{\leq} \qquad C_r\delta + \frac{\delta}{4M}\leq (C_r+1)\delta. \qedhere
\end{align}
\end{proof}

\bibliographystyle{halpha-abbrv}
\bibliography{gautam-refs1,gautam-refs2,preprints,reference,Ruiyu-refs1}

\end{document}